\documentclass{llncs}

\usepackage{color}
\usepackage{mathrsfs}
\usepackage{amsmath}
\usepackage{amssymb}
\usepackage{hyperref}
\usepackage[normalem]{ulem}
\usepackage{comment}
\usepackage[linesnumbered,ruled,titlenumbered,vlined]{algorithm2e}

\spnewtheorem{assumption}[theorem]{Assumption}{\bfseries}{\itshape}

\newcommand{\word}[1]{\ensuremath{\boldsymbol{#1}}}
\newcommand{\bfa}{\word{a}}
\newcommand{\bfb}{\word{b}}
\newcommand{\bfc}{\word{c}}
\newcommand{\bfe}{\word{e}}

\newcommand{\bfg}{\word{g}}
\newcommand{\bfh}{\word{h}}

\newcommand{\bfm}{\word{m}}

\newcommand{\bft}{\word{t}}

\newcommand{\bfv}{\word{v}}

\newcommand{\bfx}{\word{x}}
\newcommand{\bfy}{\word{y}}

\newcommand{\bfeta}{\word{\eta}}
\newcommand{\bfzero}{\word{0}}

\newcommand{\mat}[1]{\ensuremath{\boldsymbol{#1}}}
\newcommand{\bfA}{\mat{A}}
\newcommand{\bfB}{\mat{B}}
\newcommand{\bfC}{\mat{C}}
\newcommand{\bfD}{\mat{D}}
\newcommand{\bfE}{\mat{E}}
\newcommand{\bfF}{\mat{F}}
\newcommand{\bfG}{\mat{G}}
\newcommand{\bfGt}{\bfG_{\ensuremath{\text{\bf T}}}}
\newcommand{\Gpub}{\bfG_{\ensuremath{\text{pub}}}}
\newcommand{\Gsec}{\bfG_{\ensuremath{\text{sec}}}}
\newcommand{\Gisec}{\tclosure{\bfG}{i}_{\ensuremath{\text{sec}}}}
\newcommand{\Gtpub}{\bfG_{\ensuremath{\text{Tpub}}}}
\newcommand{\bfH}{\mat{H}}
\newcommand{\Hpub}{\bfH_{\ensuremath{\text{pub}}}}
\newcommand{\bfI}{\mat{I}}
\newcommand{\bfM}{\mat{M}}
\newcommand{\bfP}{\mat{P}}
\newcommand{\bfQ}{\mat{Q}}
\newcommand{\bfR}{\mat{R}}
\newcommand{\bfS}{\mat{S}}
\newcommand{\bfT}{\mat{T}}
\newcommand{\bfU}{\mat{U}}
\newcommand{\bfV}{\mat{V}}
\newcommand{\bfX}{\mat{X}}
\newcommand{\bfY}{\mat{Y}}
\newcommand{\bfZ}{\mat{Z}}
\newcommand{\Moore}[2]{\mathbf{M}_{#1}(#2)}

\renewcommand{\ker}{{\rm Ker}}

\newcommand{\K}{\mathbb{K}}

\newcommand{\Fq}{\mathbb{F}_q}
\newcommand{\Fqm}{\mathbb{F}_{q^m}}
\newcommand{\Fqmn}{\mathbb{F}_{q^m}^n}

\newcommand{\code}[1]{\mathscr{#1}}

\newcommand{\CC}{\code{C}}
\newcommand{\DC}{\code{D}}

\newcommand{\EE}{\code{E}}

\newcommand{\GG}{\code{G}}

\newcommand{\Cpub}{\CC_{\rm pub}}
\newcommand{\Cipub}{\tclosure{\CC}{i}_{\ensuremath{\text{pub}}}}
\newcommand{\Csec}{\CC_{\rm sec}}
\newcommand{\Cisec}{\tclosure{\CC}{i}_{\ensuremath{\text{sec}}}}
\newcommand{\Ctsec}{\tclosure{\CC}{t}_{\ensuremath{\text{sec}}}}
\newcommand{\Gab}[2]{\code{G}_{#1}(#2)}

\newcommand{\TGab}[4]{\code{C}_{#1,#2,#3,#4}}

\newcommand{\Supp}{\textrm{Supp}}
\newcommand{\Rsupp}{\textrm{RowSupp}}
\newcommand{\dual}[1]{{#1}^{\perp}}
\newcommand{\RS}[3]{\ensuremath{\mathbf{RS}}_{#1}(#2)}

\newcommand{\dmin}[1]{\ensuremath{d_{\textrm{min}}}\left(#1\right)}

\newcommand{\Mspace}[2]{\mathcal{M}_{#1}(#2)}
\newcommand{\GL}[2]{\ensuremath{\mathbf{GL}}_{#1}(#2)}
\newcommand{\rk}[1][q]{\mathbf{rank}_{#1}}
\newcommand{\rank}{\mathbf{rank}}
\newcommand{\RowSp}{\mathbf{RowSp}}

\newcommand{\qpoly}{\mathcal{L}}
\newcommand{\qpolyd}[1]{\mathcal{L}^{<#1}}
\newcommand{\qpolydd}[1]{\mathcal{L}^{\leq#1}}
\newcommand{\qtwistedpoly}[5]{\mathcal{L}^{#1,#2}_{#3,#4,#5}}
\newcommand\smvee{\hbox{$\scriptscriptstyle\vee$}}
\newcommand{\adj}[1]{{#1}^{\raise0.9ex\smvee}}

\newcommand{\qdeg}{\deg_q}

\newcommand{\eqdef}{\stackrel{\textrm{def}}{=}}
\renewcommand{\span}{\mathbf{Span}}
\renewcommand{\leq}{\leqslant}

\renewcommand{\geq}{\geqslant}

\newcommand{\map}[4]{\left\{
    \begin{array}{ccc}
      #1 & \longrightarrow & #2\\
      #3 & \longmapsto & #4
    \end{array}\right.
}
\newcommand{\ie}{{\em i.e. }}

\newcommand{\prob}{\mathbf{Prob}}

\newcommand{\stabr}[1]{\textrm{Stab}_{\textrm{right}}(#1)}
\newcommand{\cond}[2]{\textrm{Cond}(#1, #2)}
\newcommand{\tclosure}[2]{\overline{#1}^{#2}}

\author{Alain Couvreur\inst{1,2} \and Ilaria
  Zappatore\inst{3}\thanks{The authors are funded by the French Agence
    Nationale de la Recherche through the France 2023 ANR project
    ANR-22-PETQ-0008 PQ-TLS and the ANR-21-CE39-0009-BARRACUDA.}}
\institute{Inria \and
  LIX, CNRS UMR 7161, \'Ecole Polytechnique,\\
  Institut Polytechnique de Paris,\\
  1 rue Honor\'e d'Estienne d'Orves\\
  91120 {\sc Palaiseau Cedex}, {\sc (France)}\\
  \email{alain.couvreur@inria.fr }\and
  XLIM, CNRS UMR 7252, Universit\'e de Limoges\\
  123, avenue Albert Thomas\\
  87060 {\sc Limoges Cedex}, {\sc (France)}\\
  \email{ilaria.zappatore@unilim.fr} }

\begin{document}
\title{An extension of Overbeck's attack with an application to
  cryptanalysis of Twisted Gabidulin-based schemes}
\maketitle
\begin{abstract}
  In this article, we discuss the decoding of Gabidulin and
  related codes from a cryptographic point of view, and we observe that
  these codes can be decoded solely from the knowledge of a generator
  matrix.  We then extend and revisit Gibson and Overbeck attacks
  on the generalized GPT encryption scheme (instantiated with the
  Gabidulin code) for different ranks of the distortion matrix.  We apply our attack to the case of an instantiation with twisted Gabidulin
  codes.
\end{abstract}



\section*{Introduction}
The most promising post-quantum alternatives to RSA and elliptic
curve cryptography are based on error--correction based paradigms. The
metric which quantifies the amount of noise, can be either Euclidean
(lattice--based cryptography), Hamming (code--based cryptography) or
the rank metric.  The latter has been much less investigated than the
first two.  However, it offers an interesting range of
primitives with rather short keys
\cite{ABDGHRTZABBBO19,AABBBBCDGHZ20,ABGHZ19}. In addition, the
Gabidulin code family benefits from a decoder that corrects any error
up to a fixed threshold. This makes it possible to design schemes with
a zero failure rate, such as RQC \cite{AABBBBCDGHZ20}.  Although no
rank--based submission was selected for standardization, NIST
encouraged the community to continue the research efforts in the design and security of rank--metric based
primitives.

Historically, the first primitive based on rank metric was proposed by
Gabidu\-lin, Paramonov and Tretjakov \cite{GPT91}. It was a
McEliece--like scheme where the structure of a Gabidulin code is
hidden. This scheme was first attacked in exponential time by Gibson
\cite{G95,G96}. Then, Gabidulin and Ourivski proposed an improvement
of the system that was resistant to Gibson's attack \cite{GO01,OG03}. Later, Overbeck
\cite{O05,O08} proposed a polynomial time attack which breaks both GPT and
its improvements. Gabidulin {\em et. al.} then introduced several
variants of the GPT based on a different column scrambler $\bfP$, so
that some entries of $\bfP^{-1}$ can be in $\Fqm$
\cite{G08,GRH09,RGH11}. However, in \cite{OTN16} the authors proved
that for all of the aforementioned versions, the shape of the public key
is in fact unchanged and remains subject to Overbeck--type attacks.

The natural approach to circumvent Overbeck's attack is to replace the
Gabidulin codes with another family equipped with an efficient decoding
algorithm. However, only a few such families exist. On the one
hand, there are the LRPC codes \cite{GMRZ13} which lead to the ROLLO
scheme \cite{ABDGHRTZABBBO19}.  On the other hand, one can in a way
deteriorate the structure of the Gabidulin codes, at the cost of a
loss of efficiency of the decoder. Loidreau \cite{L10} proposed to
encrypt with a Gabidulin code perturbed by some $\Fqm$--linear
operation. This proposal was subject to polynomial--time
attacks for the smallest parameters \cite{CC20,G22,BL22}, while,  for
larger parameters, it remains secure so far. In another direction,
Puchinger {\em et. al.} \cite{PRW18} proposed to replace Gabidulin
codes with twisted Gabidulin codes. However their proposal was only
partial, since they could not provide an efficient decoder correcting up
to half the minimum distance.

\subsubsection*{Our contributions.}
The contribution we make in this article is threefold.

\smallskip

First, we discuss the decoding of Gabidulin codes and twisted Gabidulin
codes. Using the result of \cite{BC21}, we explain how to correct
errors for such codes without always being able to correct up to half of
the minimum distance.
From a cryptographic point of view, we highlight an important
observation: if in Hamming metric, decoding Reed--Solomon
codes requires the knowledge of the evaluation sequence, in the rank metric,
Gabidulin codes can be decoded solely from the knowledge of a
generator matrix. This observation extends to twisted
Gabidulin codes as soon as the decoding radius is below a certain
threshold.

\smallskip

Second, we revisit the Overbeck's attack and propose an extension.  Specifically, from a public code $\Cpub$, the original Overbeck's attack
is based on the computation of
$\Lambda_i (\Cpub) = \Cpub + \Cpub^q + \cdots + \Cpub^{q^i}$.  For the
attack to succeed, a trade--off on the parameter $i$ must be
satisfied. On the one hand, $i$ must be large enough to rule out the
random part (called {\em distortion matrix}) in $\Cpub$ used to mask
the hidden code. On the other hand, $i$ must be small enough so that
$\Lambda_i (\Cpub)$ does not to fill in the ambient space. In the
this article, we propose an extension of the Overbeck's attack that
limits our goal to the smallest possible $i$, namely $i = 1$. This
relaxation is based on calculations on a certain automorphism algebra
of the code $\Lambda_1(\Cpub)$ and extends the range of the attack.

\smallskip

Third, we investigate in depth the behavior of twisted Gabidulin codes
with respect to the $\Lambda_i$ operator.

\smallskip

The aforementioned contributions lead to an attack on a variant of GPT
proposed by Puchinger, Renner and Wachter--Zeh \cite{PRW18}.  In this
variant, the authors used two techniques to resist Overbeck's
attack. First, they mask the code with a {\em distortion matrix} of
very low rank. Second, they replace Gabidulin codes with twisted
Gabidulin codes. The authors chose twisted Gabidulin codes $\CC$ so
that for any positive $i$, the code $\Lambda_i(\CC)$ may never have
co-dimension $1$ (see \cite[Theorem 6]{PRW18}).
In this article, we prove that the latter property is not a strong enough security assessment for twisted
Gabidulin codes and that
the aforementioned contributions lead directly to an attack on the
Puchinger {\em et. al.}'s variant of GPT.

\subsubsection*{Outline of the article.}
The article is organized as follows. Section~\ref{sec:preliminaries}
introduces some basic notations used in this paper, as well as
Gabidulin codes, their twisted version and the GPT cryptosystem. In
Section~\ref{sec:decoding} we first discuss the decoding of
Gabidulin codes and propose an algorithm
(Algorithm~\ref{algo:gabi_dec}), which does not need to know the
evaluation sequence. We then explain how to decode twisted
Gabidulin codes, under a certain decoding radius.  In
Section~\ref{sec: overbeck}, we revisit the Overbeck's attack on the
GPT scheme instantiated with Gabidulin codes and we make some remarks
on the structure of the generator matrix of the code obtained by
applying the $q$-sum operator to the public key.  In
Section~\ref{sec:extension}, we propose an extension of the Overbeck's
attack to the GPT scheme instantiated with either Gabidulin or
twisted Gabidulin codes. Finally, in Section~\ref{sec:dontTwist}
we examine the behavior of the $q$-sum operator applied to the public
key of the GPT system instantiated with twisted Gabidulin codes. We
then show that we can exploit the structure of its generator matrix to
break the corresponding scheme using either the Overbeck's attack, or
more generally, its previously proposed extension.


\section{Prerequisites}\label{sec:preliminaries}
In this section we introduce the basic notions we will use throughout the
paper, starting with the notations used. Then, we briefly
introduce the Gabidulin codes and their twisted version, and finally the
GPT cryptosystem.

\subsection{Notation}

Let $q$ be a prime power, $\Fq$ be a finite field of order $q$, and
$\Fqm$ be the extension field of $\Fq$ of degree $m$.  In this article, vectors are represented by lowercase bold letters: $\bfa, \bfb, \bfx$,
and matrices by uppercase bold letters $\bfM, \bfG, \bfH$. We also
denote the space of $m \times n$ matrices with entries in a general
field $\K$, by $\Mspace{m,n}{\K}$. In the square case, \ie $m=n$, we
simplify the notation by writing $\Mspace{n}\K$, and we denote by
$\GL{n}{\K}$ the group of $n \times n$ invertible matrices.

\subsection{Rank metric codes}
Given $\bfx=(x_1, \ldots, x_n)$ a vector in $\Fqmn$, we can define its
{\em support} as,
\[
\Supp(\bfx)\eqdef\span_{\Fq}\{x_1, \ldots, x_n\}
\]
and 
\[
\rk(\bfx)\eqdef\dim(\Supp(\bfx)).
\]
The \emph{rank distance} (briefly \emph{distance}) of two vectors
$\bfx, \bfy\in \Fqmn$ is
\[
	\textrm{d}(\bfx, \bfy)\eqdef\rk(\bfx-\bfy).
  \]  
A \emph{rank metric code} $\CC$ of \emph{length} $n$ and
\emph{dimension} $k$ is an $\Fqm$-vector subspace of $\Fqmn$. Its
\emph{minimum distance} is defined as,
\[
\textrm{d}_{min}(\CC)\eqdef\min_{\bfx \in \CC\setminus\{0\}}\{\rk(\bfx)\}.
\]

By choosing an $\Fq$-basis $\mathcal B$ of $\Fqm$, any codeword
$\bfc\in \CC$ can be written as a matrix
$\bfM_{\mathcal{B}}(\bfc) \in \Mspace{m,n}{\Fq}$ by representing any
element $c_i\in \Fqm$ as a column vector whose entries are its
coefficients in the basis $\mathcal B$.  With this point of view, one
can introduce a second notion of support which is less considered in
the literature but will be useful in the sequel.

\begin{definition}
  The {\em row support} $\Rsupp (\bfc)$ of a vector $\bfc \in \Fqm^n$
  is the row span of the $m \times n$ matrix
  $\bfM_{\mathcal B}(\bfc)$.
\end{definition}

Note that the row support of a vector is an intrinsic notion that does
not depend on the choice of the basis $\mathcal B$.  Moreover, as for
the support, the rank of a vector equals its row support.

\begin{remark}
One could have defined rank metric codes as spaces of
  matrices endowed with the same rank metric. Such a framework is more
  general than ours since a matrix subspace of $\Mspace{m,n}{\Fq}$
  is not $\Fqm$--linear in general. But considering such
  rank metric codes would be useless in what follows.
\end{remark}

Two codes $\CC, \DC \subseteq \Fqm^n$ are said to be {\em right
  equivalent} if there exists $\bfP \in \GL{n}{\Fq}$ such that for any
$\bfc \in \CC$, $\bfc \bfP \in \DC$. We denote this as
``$\CC \bfP = \DC$''. We emphasize that $\bfP$ should have its entries in
$\Fq$ and {\bf not} in $\Fqm$. In this way, the map
$\bfx \mapsto \bfx \bfP$ is rank--preserving, {\em i.e.} is an
isometry with respect to the rank metric.

Finally, the \emph{dual} $\dual{\CC}$ of a code $\CC\in \Fqmn$ is the
\emph{orthogonal} of $\CC$ with respect to the canonical inner product
in $\Fqm$,
\[
\map{\Fqm \times \Fqm}{\Fqm}{(\bfx, \bfy)}{\sum_{i=1}^n x_iy_i.}
\]
We frequently apply the \emph{component-wise Frobenius
  map} to vectors and codes:, given
$\bfc=(c_1, \ldots, c_n)\in \Fqmn$ and $0\leq i \leq m-1$, we denote
\[
\bfc^{[i]}\eqdef(c_1^{q^i}, \ldots, c_n^{q^i}).
\]
Given an $[n,k]$ code $\CC\subset \Fqmn$, we write
\[
\CC^{[i]}\eqdef\{\bfc^{[i]} \mid \bfc\in \CC\}.
\]
We also define the ($i$\emph{-th}) $q$-\emph{sum} of $\CC$ as,
\[
\Lambda_i(\CC)\eqdef\CC+\CC^{[1]}+\dots + \CC^{[i]}.
\]
We notice that if $\bfG \in \Mspace{k,n}{\Fqm}$ is a generator matrix
of $\CC$, the matrix
\begin{equation}\label{eq:lam(C)}
\begin{pmatrix}
\bfG\\
\bfG^{[1]}\\
\vdots\\
\bfG^{[i]}
\end{pmatrix}\in \Mspace{(i+1)k, n}{\Fqm}
\end{equation}
is a generator of the $q$-sum of $\CC$, \ie $\Lambda_i(\CC)$.  By
abuse of notation we sometimes denote the matrix of \eqref{eq:lam(C)} as
$\Lambda_i(\bfG)$.

\subsection{Gabidulin codes}
$q$-\emph{polynomials} were first introduced in \cite{O33}. They are
defined as $\Fqm$-linear combinations of the monomials
$X, X^q, X^{q^2},\dots, X^{q^i}, \dots$ respectively denoted by
$X, X^{[1]}, X^{[2]}, \ldots, X^{[i]}, \ldots$ Formally, a nonzero
$q$-polynomial $F$ is defined as,
\[
F = \sum_{i=0}^d f_iX^{[i]}
\]
assuming that $f_d\neq 0$. The integer $d$ is called
\emph{$q$--degree} of $F$ and we denote it $\qdeg f$.
We equip the space of $q$--polynomial with a non-commutative
algebra structure, where the multiplication law is the composition of
polynomials. In particular, the product law is given by the following relations
extended by $\Fqm$--linearity:
\[
  \forall i, j \in \mathbb N,\ \forall a \in \Fqm,
  \quad X^{[i]} X^{[j]} = X^{[i+j]}\quad \text{and}\quad
  X^{[i]} a = a^{q^i}X^{[i]}.
\]
Any $q$--polynomial $F$ induces an $\Fq$--endomorphism
$\Fqm \rightarrow \Fqm$ and the {\em rank of $F$} will be defined as
the rank of its induced endomorphism.

Denote by $\qpoly$ the ring of all $q$--polynomial
and by $\qpolyd e$ the $\Fqm$--linear space of $q$--polynomials
of $q$--degree less than $e$, namely:
\[
\qpolyd e\eqdef \{f \in \qpoly \mid \qdeg f<e\}.
\]
Given two positive integers $k,n$, with $k< n\leq m$ and
$\bfg\in \Fqmn$ of $\rk(\bfg)=n$, the \emph{Gabidulin code} of length
$n$ and dimension $k$ is defined as
\[
  \Gab{k}{\bfg}
  \eqdef\{(F(g_1), \ldots, F(g_n)) \mid F\in \qpolyd{k}\}.
\]

A \emph{generator matrix of this code} is a \emph{Moore matrix}
(see for instance \cite[\S~1.3]{G96b}), \ie a matrix of the form
\begin{equation}\label{eq:MooreMatrix}
\Moore{k}{\bfg}\eqdef\begin{pmatrix}
\bfg\\
\bfg^{[1]}\\
\vdots\\
\bfg^{[k-1]}\\
\end{pmatrix}=\begin{pmatrix}
g_1 &g_2 &\ldots &g_n\\
g_1^{q} &g_2^{q} &\ldots &g_n^{q}\\
\vdots &\vdots &\ldots &\vdots\\
g_1^{q^{k-1}} &g_2^{q^{k-1}} &\ldots &g_n^{q^{k-1}}\\
\end{pmatrix}.
\end{equation}


%
%
%
  
Gabidulin codes are \emph{Maximum Rank Distance} (MRD) codes, \ie
their minimum distance is $\dmin{\Gab{k}{\bfg}}=n-k+1$ and they benefit from a decoding
algorithm correcting up to half the minimum distance (see
\cite{L06a}).

We now recall the following classical lemmas, that will be useful in
the rest of the paper.

\begin{lemma}\label{lem:multGabByInvFq}
  Let $\Gab{k}{\bfg}$ be a Gabidulin code and $\bfT\in
  \GL{n}{\Fq}$. Then $\Moore{k}{\bfg}\bfT$ is a generator matrix of
  the Gabidulin code $\Gab{k}{\bfg\bfT}$.
\end{lemma}

In short, a right--equivalent code to a Gabidulin code is a Gabidulin code
with another evaluation sequence.

\begin{lemma}[{\cite[Theorem 7]{G85}}]
	The dual of the Gabidulin code $\Gab{k}{\bfg}$ is the Gabidulin code
	$\Gab{n-k}{\bfy^{[-n+k+1]}}$, where $\bfy$ is a nonzero vector
	in $\Gab{n-1}{\bfg}^{\perp}$.
\end{lemma}
 
%
%

\subsection{Twisted Gabidulin codes}\label{subsec:twisted} Twisted
Gabidulin codes were
first introduced in \cite{S16a} and contain a broad family of MRD codes that are not
equivalent to Gabidulin codes. The construction of these codes was
then generalized in \cite{PRS17a,PRW18}.
We consider the $q$--polynomials of the form
\begin{equation}\label{eq:twisted-q-poly}
  F = \sum_{i=0}^{k-1}
  f_iX^{[i]}+\sum_{j=1}^{\ell}\eta_jf_{h_j}X^{[k-1+t_j]} ,
\end{equation}
where the $f_i$'a are in $\Fqm$, $\ell \leq n-k$,
$\bfh\in \{0, \ldots, k-1\}^{\ell}$,
$\bft\in \{1, \ldots, n-k\}^{\ell}$ (with distinct $t_i$) and
$\bfeta\in (\Fqm^*)^{\ell}$.  We denote by
$\qtwistedpoly{n}{k}{\bft}{\bfh}{\bfeta}$ the space of all
$q$--polynomials of the form \eqref{eq:twisted-q-poly} with parameters
$\bfh,\bft, \bfeta$.  Now, given a vector $\bfg\in \Fqmn$, with
$\rk(\bfg)=n$, the $[\bfg,\bft, \bfh, \bfeta]$-\emph{twisted Gabidulin
  code} of \emph{length} $n$, dimension $k$, ${\ell}$ twists,
\textit{hook vector} $\bfh$, \textit{twist vector} $\bft$ and
evaluation sequence $\bfg$ is defined as
\[
  \TGab{\bfg}{\bft}{\bfh}{\bfeta}[n,k]\eqdef\{(F(g_1),
  \ldots,F(g_n))\mid F \in {\qtwistedpoly{n}{k}{\bft}{\bfh}{\bfeta}}\}.
\]
We observe that in \cite{S16a}, Sheekey introduced a simplified
version of these codes with just one twist, \ie
$n=m, {\ell}=1,\bfh=(0),\bft=(1)$.

\begin{assumption}\label{ass:parameters}
	Throughout this paper, according to \cite{PRW18}, we consider a
	$[\bfg, \bft, \bfh, \bfeta]$-twisted Gabidulin code with ${\ell}$ twists,
	and with the following
	parameters, \begin{itemize}
		\item $t_i\eqdef (i+1)(\delta+1)$, where
		$\delta\eqdef\frac{n-k-{\ell}}{{\ell}+1}$,
		\item $0<h_1<h_2<\ldots<h_{\ell}<k-1$ and $|h_{i}-h_{i-1}|>1$.
	\end{itemize}
	for any $i$, $1\leq i \leq {\ell}$. 
\end{assumption}
This choice is particularly relevant because it allows
us to quantify the dimension of the $q$-sum operator applied to these
codes  (see
Proposition~\ref{prop:distinguish_twisted}).

\medskip
We now observe that in general, a generator matrix of a $\TGab{\bfg}{\bft}{\bfh}{\bfeta}[n,k]$ is

\begin{equation}\label{eq: twistedGenMatrix}
\begin{pmatrix}
\bfg\\
\bfg^{[1]}\\
\vdots\\
\bfg^{[h_1-1]}\\
\bfg^{[h_1]}+\eta_1\bfg^{[k-1+t_1]}\\
\bfg^{[h_1+1]}\\
\vdots\\
\bfg^{[h_{\ell}-1]}\\
\bfg^{[h_{\ell}]}+\eta_{\ell} \bfg^{[k-1+t_{\ell}]}\\
\bfg^{[h_{\ell}+1]}\\
\vdots\\
\bfg^{[k-1]}\\
\end{pmatrix}.
\end{equation}

The decoding of twisted Gabidulin codes such as their \emph{additive
  variants} has recently been studied in
\cite{RR17,R17,L19,LK19,KL20,KLZ21}. However, in \cite{RR17} there
were proposed some algorithms which allow to decode twisted Gabidulin
codes with only one twist and $\bft = (1)$, for some special choices
of parameters. They manage to correct up to
$\lfloor\frac{n-k-1}{2}\rfloor$ errors. But their decoding up to half of
the minimum distance remains an open problem.

To the best of our knowledge, the decoding of twisted Gabidulin codes
with multiple twists, or one twist with $t_1>1$ has not been
studied in the literature. We address this point in
\S~\ref{sec:decoding} for decoding radii that remain below half the
minimum distance.

\subsection{GPT system and variants}\label{subsec:classicalGPT}
The GPT cryptosystem was introduced in 1991 by {Gabidulin,
  Paramonov} and {Tretjakov} \cite{GPT91}. This system
is a \emph{rank-metric} variant of the classical
\emph{McEliece} cryptosystem \cite{M78}, in which the Goppa codes are
replaced by Gabidulin codes.  The first version of GPT was first
broken by Gibson in \cite{G95}. Gabidulin proposed a new version in
\cite{G93}, which was later attacked again by Gibson in \cite{G96}.

In this work we present the generalized version of GPT proposed by
Gabidulin and Ourivski in \cite{GO01,OG03}.  
\begin{itemize}
	\item \emph{Key Generation}. Let,
	\begin{itemize}
    \item $\Gab{k}{\bfg}$ an $[n,k]$-Gabidulin code with generator
      matrix $\Gsec$ (as in \eqref{eq:MooreMatrix});
    \item $\bfS$ a random invertible matrix in $\Mspace{k}{\Fqm}$,
    \item $\bfX$ a random matrix in $\Mspace{k,\lambda}{\Fqm}$ of
      fixed rank $1\leq s\leq \lambda$, called \emph{distortion
        matrix},
    \item $\bfP$ a random matrix in $\GL{n+\lambda}{\Fq}$, called
      \emph{column scrambler}.
	\end{itemize}
The \emph{secret key} is the triple,
\[
(\bfS,\Gsec,\bfP)
\]
and the \emph{public key} is,
\begin{equation}\label{eq:Gpub}
\bfG_{pub}\eqdef \bfS(\bfX~|~\Gsec)\bfP,
\end{equation}
where $(\bfX~|~\Gsec) \in \Mspace{k,n+\lambda}{\Fqm}$ denotes the matrix
whose columns are the concatenations of those of $\bfX$ and of $\Gsec$.
We denote $\Cpub$ the linear code with $\Gpub$ as generator matrix.
\item \emph{Encryption}. To encode a plaintext $\bfm\in \Fqm^k$,
  choose a random vector $\bfe\in \Fqm^{n+\lambda}$ of $\rk(\bfe)=t$, where
  $t=\lfloor \frac{n-k}{2}\rfloor$ and compute the ciphertext as,
\[
\bfc\eqdef \bfm \Gpub+\bfe.
\]
\item \emph{Decryption}. Apply the chosen decoding algorithm for
  Gabidulin codes to the last $n$ components of the vector,
\[
\bfc \bfP^{-1}=\bfm \bfS[\bfX|\Gsec]+\bfe \bfP^{-1}.
\]
Since $\bfP\in \GL{n+\lambda}{\Fq}$, then $\rk(\bfe \bfP^{-1}) = t$ and in particular, the rank (over $\Fq$) of the last $n$ rows of this matrix is at most $t$.
So, the decoder computes $\bfm \bfS$, and by inverting $\bfS$, the
initial message can be finally encrypted.
\end{itemize}

The description of the secret key as the triple $(\bfS, \Gsec, \bfP)$
is not the most relevant one when it comes to instantiating the scheme
with Gabidulin or twisted Gabidulin codes. In particular, once we know
the secret code $\Csec$ of the generator matrix $\Gsec$ and the
scrambling matrix, we are able to decode. So, the
knowledge of $\bfS$ is not relevant. 
Thus, in the following, we assume that $\Gpub$ as
\begin{equation}\label{eq:alter_Gpub}
  \Gpub = (\bfX ~|~ \Gsec) \bfP.
\end{equation}

\begin{remark}
  The previous scheme is instantiated with Gabidulin codes but can actually
  be instantiated with any code family equipped with a decoder that corrects
  up to $t$ errors.
\end{remark}

\begin{remark}
  The original GPT scheme \cite{GPT91} did not involve the
  distortion matrix $\bfX$ as it is.
  The seminal proposal was to use either a random
  generator matrix $\bfG$ of a Gabidulin code or a matrix
  $\bfG + \bfX_0$, where $\bfX_0$ had low rank. The latter version
  required to reduce the weight of the error term in the encryption
  process. In the following, we no longer consider this masking technique.
  The use of a distortion matrix
  with a column scrambler appeared only ten years later with the works
  of Ourivski and Gabidulin \cite{GO01,OG03}.
\end{remark}


\section{On the decoding of Gabidulin codes and their twists}
\label{sec:decoding}

In this section, we discuss further the decoding of Gabidulin and twisted
Gabidulin codes. We show that, although decoding twisted Gabidulin
codes up to half the minimum distance remains an open problem, their
decoding up to a smaller radius is possible, using the same decoder as
for Gabidulin codes. This approach was developed in \cite{BC21} and is
related to that of Gaborit, Ruatta and Schrek in
\cite[\S~V--VI]{GRS16}.

We begin by examining the decoding of Gabidulin codes.

\subsection{An important remark on the decoder of Gabidulin
  codes}\label{ss:key_remark}
It is well--known that the Gabidulin codes have a decoder that corrects up to half the minimum distance (see for instance
\cite{L06a}). This algorithm is analogous to the Welch--Berlekamp
algorithm for Reed--Solomon codes. An important fact from a
cryptographic point of view is that, given a Reed--Solomon
code
\[
  \RS{}{k}{\bfx} \eqdef \left\{(f(x_1), \dots, f(x_n)) \mid
  f \in \Fq[X],\ \deg f < k\right\},
\]
where $\bfx = (x_1, \dots, x_n) \in \Fq^n$ has distinct entries, the
knowledge of the vector $\bfx$ is necessary to run the decoding
algorithm.  However, given a Gabidulin code $\Gab{k}{\bfg}$, it is
possible to decode without knowing $\bfg$.  Indeed,
given as input $\bfy = \bfc + \bfe$ where $\bfc \in \Gab{k}{\bfg}$ and
$\rk(\bfe) \leq t \eqdef \frac{n-k}{2}$, the decoding algorithm first
consists in finding a $q$--polynomial $P(x)$ of degree at most $t$
which vanishes at the entries of $\bfe$. This can be done by solving
the $\Fqm$--linear system
\begin{equation}\label{eq:key_equation}
  P(\bfy) \eqdef (P(y_1), \dots, P(y_n)) \in \Gab{k+t}{\bfg}
\end{equation}
whose unknowns are the coefficients of $P \in \qpolydd{t}$. Next, the
code $\Gab{k+t}{\bfg}$ can be computed by simply knowing a
generator matrix of $\Gab{k}{\bfg}$, thanks to the following
well--known statement.
\begin{proposition}[{\cite[Lem.~5.1]{O08}}]\label{prop:lambdai(Gab)}
  Let $\bfg\in \Fqmn$, with $\rk(\bfg)=n$ and $\Gab{k}{\bfg}$ an
  $[n,k]$ Gabidulin code. Then,
\[
\Lambda_i(\Gab{k}{\bfg})=\Gab{k+i}{\bfg}.
\]
In particular,
\[
\dim(\Lambda_i(\Gab{k}{\bfg}))=\min\{n,k+i\}.
\]
\end{proposition}
 

Next, for any $P$ satisfying~(\ref{eq:key_equation}), we have $P(\bfy)
= P(\bfc) + P(\bfe)$.  By construction, $P(\bfc) \in
\Lambda_t(\Gab{k}{\bfg}) = \Gab{k+t}{\bfg}$ and hence, $P(\bfe) \in
\Gab{k+t}{\bfg}$. Moreover, we have $\rk (P(\bfe)) \leq \rk (\bfe)
\leq t$, while $\Lambda_t(\Gab{k}{\bfg}) =
\Gab{k+t}{\bfg}$ has minimum distance $n-k-t+1$. Therefore, for $t
\leq \frac{n-k}{2}$, which entails $t <
n-k-t+1$, we should have $P(\bfe) = 0$ for any
$P$ satisfying \eqref{eq:key_equation}.  Thus, the kernel of
$P$ contains the support of
$\bfe$ and the knowledge of the support of the error allows to solve
the decoding problem by solving a linear system.  See for instance
\cite[\S~IV.a]{GRS16}, \cite[\S~III.A]{AGHT18}.

\begin{algorithm}
  \DontPrintSemicolon
  \KwIn{A Gabidulin code $\CC$ represented by a
    generator matrix $\bfG$, an integer $t$ and a vector
    $\bfy\in \Fqm^n$}
  \KwOut{A vector $\bfc \in \CC$ such that
    $\rk{(\bfy - \bfc)} \leq t$ if exists and `?' otherwise.}

  \medskip
  Compute
  $P\in \qpolydd{t}\setminus \{0\}$ such that
  $P(\bfy) \in \Lambda_t(\CC)$\;

  Compute (if exists) $\bfe \in \Fqm$ such that
  $\Supp(\bfe) \subseteq \ker (P)$ and
  $\bfy - \bfe \in \CC$\;
  \If{$\bfe$ exists}{Return
    $\bfy - \bfe$} \Else{Return `?'}
	\caption{Decoding algorithm of Gabidulin codes without knowing the
      evaluation sequence\label{algo:gabi_dec}}
\end{algorithm}

Algorithm~\ref{algo:gabi_dec} summarizes the previous discussion.
Note that, with the knowledge of the evaluation sequence $\bfg$, the
algorithm could be terminated by performing an Euclidean division or
using the Extended Euclidean Algorithm in the non-commutative ring
$\qpoly$ instead of using \cite[\S~IV.a]{GRS16}, \cite[\S~III.A]{AGHT18}.

The key observation here is the following : {\bf decoding a Gabidulin code
  $\Gab{k}{\bfg}$ is possible without knowing the vector
  $\bfg$.}

\begin{remark}
  In GPT original public key encryption scheme \cite{GPT91} the public
  code is a Gabidulin code with no distortion matrix. In this
  situation, the previous discussion shows that this proposal is
  immediately broken without trying to compute a description (\ie an
  evaluation sequence) of the public code.
\end{remark}


\subsection{Decoding twisted Gabidulin codes}\label{sec:decodingTwisted}
If some twisted Gabidulin codes are proven to be MRD without being
equivalent to Gabidulin codes, the question of decoding them up to
half the minimum distance remains open. For {\em twisted Reed--Solomon
  codes}, the Hamming metric analogues introduced in \cite{BPR17}, it
is shown in \cite{BBPR18} how they can be decoded up to half the minimum
distance at the cost of an exhaustive search on the terms associated with
the twists. Thus, the decoding complexity of a twisted Reed--Solomon
code with $\ell$ twists is $O(q^\ell)$ times the complexity of the
decoding of a Reed--Solomon code. This can be transposed to
twisted Gabidulin codes but the cost overhead is $O(q^{m\ell})$ times
the cost of decoding a Gabidulin code, which is exponential in $m$ and
so in the code length $n$ (since $n \leq m$).

Although one does not know how to efficiently decode twisted Gabidulin
codes up to half the minimum distance, one can apply the
Algorithm~\ref{algo:gabi_dec} to them.  Given $\bfy = \bfc + \bfe$, where
$\bfc$ is a codeword of a twisted Gabidulin code $\CC$ and
$\rk (\bfe) \leq t$ for some $t$ we will discuss later,
compute $P \in \qpoly^{\leq t}$ such that
\begin{equation}\label{eq:key_equation2}
  P(\bfy) \eqdef (P(y_1), \dots, P(y_n)) \in \Lambda_t(\CC).
\end{equation}
Such a solution $P$ satisfies $P(\bfe) \in \Lambda_t(\CC)$. The
difference with the Gabidulin case is that we do not have an {\em a priori}  lower
bound on the minimum distance of
$\Lambda_t(\CC)$.
However we have the following result.
\begin{proposition}[{\cite[Theorem~4]{PRW18}}]\label{prop:distinguish_twisted}
	Given a twisted Gabidulin code
	$\TGab{\bfg}{\bft}{\bfh}{\bfeta}[n,k]$ (where parameters are chosen
	according to Assumption~\ref{ass:parameters}), then
	\[
	\forall i \geq 0,\quad
	\dim(\Lambda_i(\TGab{\bfg}{\bft}{\bfh}{\bfeta}[n,k]))
	=\min\{k+i+\ell(i+1),n\}.
	\]
\end{proposition}
Proposition~\ref{prop:distinguish_twisted} entails that for a twisted
Gabidulin code $\CC$ with $\ell$ twists, we have
\begin{equation}\label{eq:dimLambdaTwist}
	\dim_{\Fqm} \Lambda_t (\CC) \leq k-1 + (t+1)(\ell + 1).
\end{equation}
Now, let us consider the dimension of $\Lambda_t (\EE)$.  Since
$\Lambda_{t}(\EE)$ is the image of $\qpoly^{\leq t}$ by the map
$Q \mapsto Q(\bfe)$, we have
\[
  \dim_{\Fqm} (\Lambda_{t}(\EE)) = \dim_{\Fqm} (\qpoly^{\leq t}) - \dim_{\Fqm} \{
  Q \in \qpoly^{\leq t} ~|~ Q(\bfe) = 0\}.
\]
First, $\dim (\qpoly^{\leq t}) = t+1$. Second, recall that there exists a unique monic $q$--polynomial $P$ of $q$--degree
$\rk{(\bfe)}$ such that $P(\bfe) = 0$. Therefore, 
\[
  \{ Q \in \qpoly^{\leq t} ~|~ Q(\bfe) = 0\} = \{F \circ P ~|~ F \in
  \qpoly^{\leq t-\rk (\bfe)}\}
\]
and the latter space has dimension $t - \rk (\bfe) + 1 \geq 1$.
Putting all together, we deduce that
\[
  \dim_{\Fqm} (\Lambda_t(\EE)) \leq t.
\]

We claim that if
\begin{equation}\label{eq:key_condition}
	\dim_{\Fqm} \Lambda_t (\CC) + t \leq n,
\end{equation} the spaces $\Lambda_t(\CC)$
and $\Lambda_t (\EE)$ are very likely to have a zero intersection.
The validity of this claim are given in \S~\ref{ss:evidence}.
This would entail that for any $P \in \qpoly^{\leq t}$ satisfying
\eqref{eq:key_equation2}, we have $P(\bfe) = 0$. 
Therefore, from \eqref{eq:dimLambdaTwist}  and \eqref{eq:key_condition} we can conclude that if,
\[
t \leq \frac{n-k-\ell}{\ell + 2}\cdot
\]
then we can decode twisted Gabidulin codes as classical Gabidulin
codes : form the kernel of $P$, we get the error support and finally
the error itself is deduced using \cite[\S~IV.a]{GRS16},
\cite[\S~III.A]{AGHT18}.  This decoding radius is rather pessimistic
since the dimension of $\Lambda_t(\CC)$ may be much smaller depending
on the way the twists are chosen. Therefore, the above bound is what
we can expect in the worst case.

\subsection{Discussion about the claim}\label{ss:evidence}
Suppose that the error $\bfe$ is obtained as follows: draw a
uniformly random subspace $V \subseteq \Fq^n$ of dimension $t$ and
then draw a uniformly random vector $\bfe$ among the vector with row
support contained in $V$.  One can easily prove that all the elements
of $\Lambda_t(\bfe)$ have their row support contained in $V$.

Therefore, the intersection $\Lambda_t(\EE) \cap \Lambda_t(\CC)$ consists
in elements of $\Lambda_t(\CC)$ whose row support is in $V$.
So, consider the subcode $\text{Sh}_V(\Lambda_t(\CC))$ called {\em
  shortening of $\Lambda_t(\CC)$} defined as the subcode of
$\Lambda_t(\CC)$ of vectors whose row support is contained in $V$.
This space can be obtained as follows. Consider a basis
$(\bfv_1, \dots, \bfv_{n-t})$ of the dual $V^{\perp} \subseteq \Fq^n$ of $V$
for the canonical inner product. Then, $\text{Sh}_V(\Lambda_t(\CC))$
is the kernel of the map
\[
  \map{\Lambda_t(\CC)}{\Fqm^{n-t}}{\bfc}{(\bfc \cdot \bfv_1^\top,
    \dots , \bfc \cdot \bfv_{n-t}^\top).}
\]
\begin{remark}
Note that in the above equation, $\bfc$ and the $\bfv_i$'s have
different nature, $\bfc$ has entries in $\Fqm$ while the $\bfv_i$'s
have their entries in $\Fq$.
\end{remark}

Finally, since $V$ is uniformly random, and
$\dim \Lambda_t(\CC) \leq n-t$, it is likely that the above map is
injective and hence its kernel $\text{Sh}_V(\Lambda_t (\CC))$ is
likely to be zero. Since the latter kernel contains
$\Lambda_t(\EE) \cap \lambda_t(\CC)$, we conclude that this
intersection is likely to be zero.

\subsection{A remark on the code that is actually decoded}
To conclude, let us notice an important fact for the sections to
follow.  The previously described decoder may decode a slightly larger
code than $\CC$ defined below.

\begin{definition}\label{def:ct_cond}
  Let $\CC \subseteq \Fqm^n$ be a code and $s$ be a positive
  integer. We denote by $\tclosure{\CC}{s}$ the largest code $\CC'$
  containing $\CC$ such that $\Lambda_s(\CC) = \Lambda_s(\CC')$.
\end{definition}

It is easy to check that, the aforementioned decoder actually decodes
$\tclosure{\CC}{t}$ and not only $\CC$.

\begin{remark}
  It can be proved that for a random code $\CC$ with dimension $k < \frac n s$,
  then $\CC = \tclosure{\CC}{s}$ with a high probability.
  It ca also be proved that a Gabidulin code $\CC$ of dimension $k$
  satisfies $\tclosure{\CC}{i} = \CC$ for any $i < n-k$.
\end{remark}

\begin{remark}\label{rem:alter_closure}
  An alternative definition of $\tclosure{\CC}{s}$ is given by.
  \[
    \tclosure{\CC}{s} \eqdef \bigcap_{j=0}^s {\left(\Lambda_s(\CC)\right)}^{[-j]}
  \]
\end{remark}


\section{Revisiting Overbeck's attack}\label{sec: overbeck}
In this section we revisit the Overbeck's attack of GPT instantiated with
Gabidulin codes to introduce the extension presented in
\S~\ref{sec:extension}, which will allow us to break \cite{PRW18}.

\subsection{A distinguisher}
The core of the Overbeck's attack consists in the application of the
$q$-sum operator, which allows to \emph{distinguish} Gabidulin codes from
random ones. In particular, the following proposition observes the
behavior of random codes \emph{w.r.t.} the $i$-th $q$-sum operator.

\begin{proposition}[{\cite[Prop.~1]{CC19}}]\label{prop:Lambdai(C)}
  If $\CC\subset \Fqmn$ is a $k$-dimensional random code, then for any
  $0<i<k$,
\[
 \dim(\Lambda_i(\CC))\leq \min\{n, (i+1)k\}.
\]
Moreover, for any $a \geq 0$, we have
\[
  \prob (\dim(\Lambda_i(\CC))\leq \min\{n, (i+1)k\} - a) = O(q^{-ma}).
\]  
\end{proposition} 

Gabidulin codes have a significantly different behavior with respect
to the $q$--sum compared to random codes (see
Proposition~\ref{prop:lambdai(Gab)}). In fact, we observe that if
$i<n-k$,
\[
\dim(\Lambda_i(\Gab{k}{\bfg}))=k+i<(i+1)k=\dim(\Lambda_i(\CC)), 
\]
where $\Gab{k}{\bfg}$ is a $n$-Gabidulin code of dimension $k$, and
$\CC$ is a random code, and we know from the previous proposition that
the last equality is true with high probability.

In the Overbeck's attack, the operator $\Lambda_i(\cdot)$ is used for two
related reasons.
\begin{enumerate}
\item It provides a distinguisher on the public key based on the peculiar
  behavior of Gabidulin codes with respect to $\Lambda_i (\cdot)$.
  This permits to rule out the distortion matrix \cite{O08} and to recover a decomposition of the form \eqref{eq:alter_Gpub}, in order to
  decrypt any ciphertext computed with this public key.
\item Once we have discarded the distortion matrix, we have access to the
  secret Gabidulin code and we can recover its hidden structure, \ie
  an evaluation sequence.
\end{enumerate}
We observe that the second step is not necessary since, using
Algorithm~\ref{algo:gabi_dec}, one can directly decode any message, without knowing the evaluation sequence. Thus,
in the sequel, we focus on the first step.

\subsection{The structure of
  $\Lambda_i(\Gpub)$}\label{subsec:structureLambdaiGpub}
Let $i$ be a positive integer and $\Gpub = (\bfX ~|~ \Gsec) \bfP$ a
public key as in \eqref{eq:alter_Gpub}. Recall that, in the present
section, we suppose that $\Gsec$ is a generator matrix of a Gabidulin
code. Observe that, since $\bfP\in \GL{n+\lambda}{\Fq}$, we have
$\bfP^{[i]} = \bfP$ and hence,
\begin{equation}\label{eq:Lambdai(Gpub)}
	\Lambda_i(\Gpub) = (\Lambda_i(\bfX)~|~\Lambda_i(\Gsec))\bfP.
\end{equation}

We now assume that $i<n-k$ and we focus on the matrix
$(\Lambda_i(\bfX)~|~\Lambda_i(\Gsec))$.  If we denote the distortion
matrix $\bfX$ according to its rows, \ie 
\[\bfX = \begin{pmatrix}
  \bfx_0\\
  \bfx_1\\
  \vdots\\
  \bfx_{k-1}\\
\end{pmatrix},
\]
where $\bfx_{j}\in \Fqm^{\lambda}$ for any $0\leq j \leq k-1$, then
\[
(\Lambda_i(\bfX)~|~\Lambda_i(\Gsec))=\left(\begin{array}{c|c}
	\bfx_0 &\bfg\\
	\bfx_1 &\bfg^{[1]}\\
	\vdots &\vdots\\
	\bfx_{k-1} &\bfg^{[k-1]}\\
	\hline
	\vdots & \vdots\\
	\hline
	\bfx_0^{[i]} &\bfg^{[i]}\\
	\bfx_1^{[i]} &\bfg^{[i+1]}\\
	\vdots & \vdots \\
	\bfx_{k-1}^{[i]} &\bfg^{[k-1+i]}
\end{array}\right).
\]
Now, after performing some row elimination, we finally get
\[
\left(\begin{array}{c|c}
	\bfx_0 &\bfg\\
	\bfx_1 &\bfg^{[1]}\\
	\vdots &\vdots\\
	\bfx_{k-1} &\bfg^{[k-1]}\\
	\vdots &\vdots\\
	\bfx_{k-1}^{[i]} &\bfg^{[k-1+i]}\\
	\hline
	\vdots & \vdots\\
	\hline
	\bfx_0^{[i]}-\bfx_1^{[i-1]} &\bfzero\\
	\bfx_1^{[i]}-\bfx_2^{[i-1]} &\bfzero\\
	\vdots &\vdots\\
	\bfx_{k-2}^{[i]}-\bfx_{k-1}^{[i-1]} &\bfzero\\
\end{array}\right).
\]
Thus, we have the following.

\begin{lemma}\label{lem: rkGpub}
Let $ i<n-k$. Then, up to row elimination,
\begin{equation}\label{eq:LambdaiXLambdaiG}
(\Lambda_i(\bfX)~|~\Lambda_i(\Gsec))=	\begin{pmatrix}	
		\bfX' &\Moore{k+i}{\bfg}\\
		\Lambda_{i-1}(\bfX'') &	\bfzero
	\end{pmatrix},
\end{equation}
where,
\[
\bfX'=\begin{pmatrix}
	\bfx_0\\
	\vdots\\
	\bfx_{k-1}\\
	\bfx_{k-1}^{[1]}\\
	\vdots\\
	\bfx_{k-1}^{[i]}
\end{pmatrix}
\quad
\textrm{and}
\quad
\bfX'' = \bfX^{[1]}_{\{0, \ldots, k-2\}} - \bfX_{\{1,\ldots,k-1\}}.
\]
In detail, $\bfX^{[1]}_{\{0, \ldots, k-2\}}$ is the submatrix of $\bfX^{[1]}$
composed by its first $k-1$ rows and $\bfX_{\{1,\ldots,k-1\}}$ is the
submatrix of $\bfX$ composed by its rows
starting from the second one.
\end{lemma}

We now observe that the row space of $\bfX''$, denoted
$\RowSp_{\Fqm}(\bfX'')$, is contained in the sum of the row spaces of
$\bfX$ and $\bfX^{[1]}$, which is $\RowSp_{\Fqm}(\Lambda_1(\bfX))$ and
so $\rank(\bfX'')\leq \min\{2s, \lambda\}$, where we recall that
$s=\rank(\bfX)$.

More generally,
$\RowSp_{\Fqm}(\Lambda_{i-1}(\bfX''))\subseteq
\RowSp_{\Fqm}(\Lambda_i(\bfX))$ for any $i\geq 1$. And
$\rank(\Lambda_{i-1}(\bfX''))\leq \min\{(i+1)s, \lambda\}$.

\subsection{Overbeck's attack}\label{subsec:Overbeck}
The attack consists in finding an $i<n-k$, for which
$\rank(\Lambda_{i-1}(\bfX''))=\lambda$. In this case,
\[
\dim(\Lambda_i(\GG_{pub}))=k+i+\lambda
\]
and the dimension of the dual is 
\[
\dim(\Lambda_i(\GG_{pub})^{\perp})=n-k-i.
\]
So, the code $\Lambda_i(\GG_{pub})$ admits a parity check of this form
\begin{equation}\label{eq:parityOfAttack}
	(\bfzero ~|~ \bfH_i)(\bfP^{-1})^{\top},
\end{equation} 
where $\bfH_i$ is a parity check matrix of
$\Lambda_i(\Gab{k}{\bfg})=\Gab{k+i}{\bfg}$.

After finding such an $i$, we can easily find a \textit{valid} column scrambler $\bfT\in \GL{n+\lambda}{\Fq}$,which will allow us to attack the system (see Theorem~\ref{thm:validColScramb} (\cite[Thm~5.3]{O08})).

\medskip
Therefore, {\bf the crucial part of the Overbeck's attack
  consists in finding (if there exists) a positive integer $i$,
  for which $\dim(\Lambda_{i-1}(\bfX'')) =\lambda$ and $\Lambda_i (\Csec)
  \neq \Fqm^n$ or equivalently
  $\dim(\Lambda_i(\Cpub))=\dim (\Lambda_i(\Csec)) + \lambda$.}

\begin{remark}\label{rem:PRW_assumpt}
  If for $i=n-k-1$, we have
  $\dim(\Lambda_{n-k-1}(\CC_{pub}))^{\perp}=1$, then we can perform
  the attack quite straightforwardly. Indeed, in this case there
  exists $\bfv \in \Fqmn$ which spans the entire dual.  Many papers in
  the literature describe the attack just for this choice $i$,
  claiming that we can perform it only if
  $\dim(\Lambda_{n-k-1}(\CC_{pub}))^{\perp}=1$. We stress out that {\bf this is not the
    only possible choice for $i$:} one only needs an $i<n-k$ for which
   $\Lambda_i(\CC_{pub})^\perp$ has the structure \eqref{eq:parityOfAttack}.
\end{remark}

\subsubsection{Description of the attack.}
We now briefly detail the procedure of
the attack (partially presented in the proof of \cite[Thm.~5.3]{O08}).

We know that $\Lambda_i(\CC_{pub})$ admits a parity check matrix $\Hpub$ (for simplicity, we omit the dependency on $i$) of the form \eqref{eq:parityOfAttack}. Thus, we look for
some $\bfT \in \GL{n+\lambda}{\Fq}$ for which
\begin{equation}\label{eq:T}
  \Hpub \bfT^{\top} = (\bfzero~|~\bfH')
\end{equation}
The matrix $\bfT$ is not unique. Furthermore, the following statement
taken from \cite[Thm~5.3]{O08} asserts that every invertible $\bfT$
satisfying \eqref{eq:T} is suitable to complete the attack.
For the sake of completeness, we give the proof of
this result. 

\begin{theorem}[{\cite[Thm~5.3]{O08}}]\label{thm:validColScramb}
	If there exists a positive integer $i<n-k$ for which the dimension
	of $\Lambda_i(\GG_{pub})^{\perp}$ is $n-k-i$ and if we denote by
	$\Hpub$ a generator matrix of this dual, then any
	$\bfT\in \GL{n+\lambda}{\Fq}$ such that
	\[
	\Hpub \bfT^{\top}=(\bfzero~|~\bfH')
	\]
	for some $\bfH'\in \Mspace{n-k-i, n}{\Fqm}$ is a \emph{valid
		column scrambler}, \ie there exists
	$\bfZ\in \Mspace{k,\lambda}{\Fqm}$ and $\bfg^{\star} \in \Fqm^n$
	of rank $n$, such that
	\[
	\Gpub = \bfS(\bfZ~|~\Moore{k}{\bfg^{\star}})\bfT,
	\]
	where $\Moore{k}{\bfg^\star}$ denotes the Moore matrix with
    generator vector $\bfg^\star$(see~\eqref{eq:MooreMatrix}).
\end{theorem}

\begin{proof}
	Since $\dim(\Lambda_i(\GG_{pub})^{\perp})=n-k-i$, then this dual
	admits a generator matrix of the form \eqref{eq:parityOfAttack}.
	Now, consider $\bfT\in \GL{n+\lambda}{\Fq}$ such that
	\begin{equation}\label{eq:Overbeckattackproof}
		(\bfzero~|~\bfH_i)(\bfP^{-1})^{\top}\bfT^{\top}=(\bfzero~|~\bfH')
	\end{equation}
	for some $\bfH'\in \Mspace{n-k-i, n}{\Fqm}$.
	Denote,
	\[
	\bfT\bfP^{-1}=\begin{pmatrix}
		\bfA & \bfB\\
		\bfC & \bfD\\
	\end{pmatrix}
	\]
	where $\bfA\in \Mspace{\lambda}{\Fq}$,
	$\bfB\in\Mspace{\lambda,n}{\Fq}$, $\bfC\in\Mspace{n,\lambda}{\Fq}$
	and $\bfD\in \Mspace{n}{\Fq}$.
	From \eqref{eq:Overbeckattackproof},we have that
	\[
	\bfH_i \bfB^{\top}=0 \Longrightarrow \bfB=0.
	\]
	Since $\bfP \bfT^{-1}$ is invertible, this entails in particular that
	$\bfA \in \GL{\lambda}{\Fq}$ and $\bfD \in \GL{n}{\Fq}$.
	Then, we have that
	\[
	(\bfT\bfP^{-1})^{-1} = \bfP\bfT^{-1}=\begin{pmatrix}
		\bfA^{-1} &\bfzero\\
		-\bfD^{-1}\bfC\bfA^{-1} &\bfD^{-1}
	\end{pmatrix}
	\]
	and so we get,
	\[
	\Gpub \bfT^{-1} = \bfS(\bfX~|~\Moore{k}{\bfg})\bfP\bfT^{-1} =
	\bfS(\bfZ~|~\bfG')
  \]
  for some matrix $\bfZ$, where $\bfG'$ is a generator matrix of
  $\Gab{k}{\bfg}\bfD^{-1}$, which also equals
  $\Gab{k}{\bfg \bfD^{-1}}$ since $\bfD$ is nonsingular with entries
  in $\Fq$ (see Lemma~\ref{lem:multGabByInvFq}). \qed
\end{proof}

In order find such a $\bfT$, we compute the space of the matrices
$\bfT \in \Mspace{n+\lambda}{\Fq}$ such that the $\lambda$ leftmost
columns of $\Hpub\bfT^{\top}$ are zero. Then , we need to extract
a nonsingular matrix from this solution space. This last step can
be done by picking random elements in this space until
we find a nonsingular matrix.

Once such a column scrambler $\bfT$ is computed, we can compute
$\bfc \bfT^{-1}$ and remove the leftmost $\lambda$ entries. By
Theorem~\ref{thm:validColScramb}, each of these $\bfT$'s is a valid
column scrambler and it suffices to apply the Gabidulin codes decoder
to the former vector to recover the plaintext. Recall that, from
\S~\ref{ss:key_remark}, the decoder works independently on the
knowledge of $\bfg$.

\subsection{Analyzing the dimension of $\Lambda_i (\Cpub)$ for small
  $i$'s.}\label{subsec:attackForSmallI}
In this section we study
what happens if we apply the $q$-sum operator to the public key for
small $i$'s, namely $i=1$. In particular, we will see that in this
case we can always attack the system by applying either strategies
described in \S~\ref{sec:extension} or the classical Overbeck attack.



First, we recall that, by Lemma~\ref{lem: rkGpub}, the matrix
$(\Lambda_1(\bfX)~|~\Lambda_1(\bfG))$ (see \eqref{eq:LambdaiXLambdaiG})
can be transformed into a matrix
\begin{equation}\label{eq:Lambda1XLambda1G}
\begin{pmatrix}
	\bfX' &\Moore{k+1}{\bfg}\\
	\bfX'' &\bfzero
\end{pmatrix}
\end{equation}
In this case, $\rank(\bfX'')\leq \min\{2s, \lambda\}$, where $s=\rank(\bfX)$. 
We now introduce the following useful lemma.
\begin{lemma}{\label{lem:LambdaiGpubblocks}}
	If $k\geq 4s+1$, then, up to row multiplications, 
		\begin{equation}{\label{eq:LambdaiBlock}}
[\Lambda_1(\bfX)~|~\Lambda_1(\bfG)]=	\begin{pmatrix}
		\bfzero &\Moore{k+1}{\bfg}\\
		\bfX''& \bfzero
	\end{pmatrix}
	\end{equation}
with a high probability.	
\end{lemma}
\begin{proof}
  We need to prove that
  $\RowSp_{\Fqm}(\bfX')\subseteq\RowSp_{\Fqm}(\bfX'')$.
	We first claim that
    $\RowSp_{\Fqm}(\bfX'')=\RowSp_{\Fqm}(\Lambda_1(\bfX))$ with a high
    probability. We consider the submatrix of $\bfX''$ in
    $\Mspace{\lfloor\frac{k-1}{2}\rfloor, \lambda}{\Fqm}$ obtained by selecting alternate rows of $\bfX''$. This is a 
    uniformly random matrix in
    $\Mspace{\lfloor\frac{k-1}{2}\rfloor, \lambda}{\Fqm}$. By the
    assumption $\frac{k-1}{2}\geq 2s$, it has rank equal to
    $\min\{2s,\lambda\}=\rank(\Lambda_1(\bfX))$ with a high
    probability (by Proposition~\ref{prop:Lambdai(C)}). Thus,
    $\rank(\bfX'')\geq \rank(\Lambda_1(\bfX))$ with a high probability
    and so the claim follows.
	The result derives from remarking that
    $\RowSp_{\Fqm}(\bfX')\subseteq\RowSp_{\Fqm}(\Lambda_1(\bfX))$. \qed
\end{proof}
We remark that if $\rank(\bfX)=s\geq \lambda/2$, then
$\rank(\bfX'')=\lambda$ with high probability and so we can apply
straightforwardly the Overbeck's attack (\S~\ref{sec: overbeck}).  One
could then think that it suffices to take a sufficiently small $s$ in
order to repair the system.  In the following section we show that
thanks to the structure of the matrix \eqref{eq:LambdaiBlock}, we can
construct an attack, which is an extension of the Overbeck's one,
which allows us to break the system independently from the rank of the
distortion matrix, even for the twisted Gabidulin GPT scheme.

\begin{remark}
  The condition $k \geq 4s+1$ required in
  Lemma~\ref{lem:LambdaiGpubblocks}, yields a range of parameters for which we can assert the validity of the
  result. Nevertheless, it is probably highly conservative and one could expect
  result to hold for smaller $k$ or equivalently larger $s$.
\end{remark}

\subsection{Puchinger, Renner and Wachter--Zeh variant of
  GPT}\label{subsec:PRW}
In \cite{PRW18}, the authors use simultaneously two distinct
techniques in order to resist to Overbeck's attack:
\begin{enumerate}
\item they impose the distortion matrix to have a low rank ({\em e.g.} $s = 1$ or $2$),
\item they replace Gabidulin codes by twisted ones (with parameters specified in Assumption~\ref{ass:parameters}).
\end{enumerate}

The rationale behind the use of twisted Gabidulin codes is that, one
step of Overbeck's attack consists in obtaining $\Lambda_{n-k-1}(\Csec)$
where $\Csec$ is the hidden Gabidulin code. Then the dual
$\Lambda_{n-k-1}(\CC)$ has dimension $1$ and immediately provides the
evaluation sequence. Based on this observation, the authors select
parameters for twisted Gabidulin codes such that none of the
$\Lambda_i(\Csec)$'s for $i>0$ may have codimension $1$ (see
\cite[Thm.~6]{PRW18}).

\begin{table}
\caption{Parameters from \cite{PRW18}}
\centering
\begin{tabular}{ | c | c | c | c | c | c | c |}
	\hline
	$\hspace{0.3cm}q \hspace{0.3cm}$ & $\hspace{0.3cm}k\hspace{0.3cm}$ & $\hspace{0.3cm}n\hspace{0.3cm}$ & $\hspace{0.3cm} m \hspace{0.3cm}$ & $\hspace{0.3cm}\ell \hspace{0.3cm}$ &$\hspace{0.3cm}\lambda\hspace{0.3cm}$ &$\hspace{0.3cm}s\hspace{0.3cm}$\\
	\hline
	\hline
	$2$ &$18$ &$26$ &$104$ &$2$ &$6$ &$1$\\
	\hline
	$2$ &$21$ &$33$ &$132$ &$2$ &$8$ &$1$\\
	\hline
	$2$ &$32$ &$48$ &$192$ &$2$ &$12$ &$2$\\
	\hline 

\end{tabular}
\label{tab:twistedPar}
\end{table}

As mentioned in Remark~\ref{rem:PRW_assumpt}, the choice of
computing $\Lambda_{n-k-1}(\Csec)$ is only technical and can be
circumvented in many different ways. In fact, once the distortion
matrix $\bfX$ is discarded, we can access to $\Csec$ and, using the
discussion in \S~\ref{sec:decodingTwisted}, just knowing this
code is generally enough  to decode.  However, their approach
presents another difficulty for the attacker if one wants to apply
Overbeck's attack. Indeed, the proposed parameters consider a
distortion matrix of low rank, {\em e.g.} $s=1$ or $2$ (see
Table~\ref{tab:twistedPar}). Then, to get for $\Lambda_i(\Cpub)$ a
generator matrix of the form (\ref{eq:LambdaiBlock}) with
$\Lambda_{i-1}(\bfX'')$ of full rank, one needs $i$ to be large, while
the dimensions of the $\Lambda_i (\Csec)$ increase faster than for a
Gabidulin code. Thus, for some parameters it is possible that the
computation of the successive $\Lambda_i(\Cpub)$ provide successive
codes with generator matrices of the form (\ref{eq:LambdaiBlock}), so
that $\Lambda_i(\Csec)$ becomes the full code $\Fqm^n$ before
$\Lambda_{i-1}(\bfX'')$ reaches the full rank $\lambda$.  The core of our
extension in \S~\ref{sec:extension} is the observation that there is no need for $\bfX''$ to have full
rank to break the scheme.

\begin{example}\label{ex:param_PRW18}
  According to the Table~\ref{tab:twistedPar}, suppose that
  $n = 26$, $k=18$, $\lambda = 6$ and
  $s = 1$. Then, for $\bfX''$ to have full rank $\lambda = 6$, while
  $\bfX$ has rank 1, we need to compute $\Lambda_6(\Cpub)$.
  But since the secret code has dimension $18$ and it is a twisted Gabidulin
  code, we deduce that $\dim \Lambda_6(\Csec) \geq 26$ and, since $n=26$,
  this code is nothing else than $\Fqm^{26}$. Thus, for such parameters, we cannot apply the
  Overbeck's attack.
  In fact, even if instantiated with a Gabidulin code, the Overbeck's
  attack would fail for such parameters.
\end{example}

\section{An extension of Overbeck's attack}\label{sec:extension}
As explained earlier, Overbeck's technique consists in applying the
$q$--sum operator $\Lambda_i$ to the public code, for an $i$ such that the
public code has a generator matrix of the form
\begin{equation}\label{ex:strong_shape}
  \begin{pmatrix}
    \bfI_{\lambda} & \bfzero \\ \bfzero & \Lambda_i(\Gsec)
  \end{pmatrix} \bfP,
\end{equation}
where $\Lambda_i(\Csec) \neq \Fqm^n$.  This entails that the dual code
has a generator matrix of the form
\begin{equation}\label{eq:zeroblock}
  \left(\bf0 ~|~ \bfH \right) {(\bfP^{-1})}^{\top},
\end{equation}
where $\bfH^{\top}$ is a parity--check matrix of $\Lambda_i(\Csec)$.
Then, a valid column scrambler can be computed by solving a linear system.
The point of this section is to prove that one can relax the
constraint on $i$ and only expect $\Lambda_i(\Gpub)$ to have a
generator matrix ``splitting in two blocks'', \ie{}
\begin{equation}\label{eq:splitting_shape}
  \begin{pmatrix}
     \bfY & \bf0 \\ \bf0 & \Lambda_i(\Gsec)
  \end{pmatrix}\bfP,
\end{equation}
without requiring $\bfY$ to have full rank $\lambda$.

Note that the above-described setting is precisely what happens to
$\Lambda_1(\Gpub)$ when $s=\rank (\bfX) < \lambda/2$, see
\S~\ref{subsec:attackForSmallI}, Example~\ref{ex:param_PRW18} or
\S~\ref{subsec:attackSmalliTw}.

\begin{example}
  Back to Example~\ref{ex:param_PRW18}, for such parameters, even
  instantiated with a Gabidulin code, the Overbeck's attack fails
  because there is not any $i>0$ which gives a matrix of the shape
  (\ref{ex:strong_shape}).  However, under some assumptions on the
  parameters of the code, it is likely that $\Lambda_1 (\Gpub)$ has a
  generator matrix of the shape (\ref{eq:splitting_shape}).See for
  instance Lemmas~\ref{lem:LambdaiGpubblocks}
  and~\ref{lem:lambdaiGtpub}.
\end{example}

\subsection{Sketch of the attack}
Now, let us explain how to find the hidden splitting structure
(\ref{eq:splitting_shape}) without any knowledge of the
scrambling matrix $\bfP$. Assume that $\Lambda_i (\Cpub)$
has a generator matrix of the form
\begin{equation}\label{eq:gm_Lambda_i}
  \begin{pmatrix}
    \bfY & \bf0\\
    \bf0 & \bfG_{i}
  \end{pmatrix}\bfP,
\end{equation}
where $\bfY$ is a matrix with $\lambda$ columns, $\bfG_i$ is a
generator matrix of $\Lambda_i (\Csec)$ and $\Csec$ is the hidden
code of dimension $k$. The code $\Csec$ could be either a Gabidulin code in
the case of classical GPT or a twisted Gabidulin code (see respectively
\S~\ref{subsec:classicalGPT} and the beginning of
\S~\ref{sec:dontTwist}). 




The idea consists in computing the {\em right stabilizer algebra}
of $\Lambda_i (\Cpub)$:
\[
  \stabr{\Lambda_i(\Cpub)} \eqdef \{\bfM \in \Mspace{n+\lambda}{\Fq}
  ~|~ \Lambda_i(\Cpub) \bfM \subseteq \Lambda_i(\Cpub)\}.
\]
This algebra can be computed by solving a linear system (see
\S~\ref{subsec:algebra}).  It turns out that it contains two
peculiar matrices, namely:
\begin{equation}\label{eq:idempotents}
\bfE_1  = \bfP^{-1}
\begin{pmatrix}
  \bfI_{\lambda} & \bf0 \\
  \bf0 & \bf0
\end{pmatrix} \bfP
\qquad \text{and}
\qquad
\bfE_2 = 
\bfP^{-1}
\begin{pmatrix}
  \bf0 & \bf0 \\ \bf0 & \bfI_n
\end{pmatrix} \bfP.
\end{equation}
The core of the attack consists in computing these two matrices, or
more precisely conjugates of these matrices, and then consider the code
$\Cpub \bfE_2$ which is somehow right equivalent to $\Csec$. In
particular, the right multiplication by $\bfE_2$ will annihilate the
distortion matrix $\bfX$. Let us now present the approach in more detail.

\subsection{Some algebraic preliminaries}\label{subsec:algebra}

\subsubsection{Split and indecomposable codes.}
The first crucial notion is that of {\em split} or {\em decomposable} codes.
\begin{definition}
  A code $\CC \subseteq \Fqm^n$ of dimension $k$ is said to {\em split} if
  it has a generator matrix of the form
  \[
    \begin{pmatrix}
      \bfG_1 & \bf0 \\  \bf0 & \bfG_2
    \end{pmatrix}\bfQ,
  \]
  for some matrices
  $\bfG_1 \in \Mspace{a,b}{\Fqm}, \bfG_2 \in \Mspace{k-a, n-b}{\Fqm}$
  and $\bfQ \in \GL{n}{\Fq}$.
  If no such block--wise decomposition exists, then the code is said to
  be {\em indecomposable}.
\end{definition}

\begin{remark}
Considering the code as a space of matrices, being split means that
  the code is the direct sum of two subcodes whose row supports ({\em
    i.e.} the sum of the row spaces of their elements) are in direct
  sum. This is the rank metric counterpart of Hamming codes which are
  the direct sum of two subcodes with disjoint Hamming supports.
  Note that this property is very rare and corresponds to somehow very
  {\em degenerated} codes.
\end{remark}

\subsubsection{Stabilizer algebras and conductors.}
We now define the notions that we will use throughout this section. Stabilizers are
useful invariants of codes, also called {\em idealizers} in the
literature. Conductors, are used for instance in \cite{CDG20}
and have often been used in cryptanalysis of schemes based on
algebraic Hamming metric codes, for instance \cite{COT17,CMP17,BC18}.

\begin{definition}\label{def:cond}
  Let $\CC\subseteq \Fqm^{n_1}$ and $\DC \subseteq \Fqm^{n_2}$
  be two $\Fqm$--linear codes of respective length $n_1, n_2$.
  The {\em conductor of $\CC$ into $\DC$} is defined as:
  \[
    \cond{\CC}{\DC} \eqdef \left\{ \bfA \in \Mspace{n_1, n_2}{\Fq}
      ~|~ \forall \bfc \in \CC,\ \bfc \bfA \in \DC.
    \right\}
  \]
  It is an $\Fq$--vector subspace of $\Mspace{n_1, n_2}{\Fq}$.
  Moreover, when $\CC = \DC$, then the conductor is an algebra
  which is usually called {\em right stabilizer} or {\em right idealizer}
  of $\CC$ and denoted
  \[
    \stabr{\CC} \eqdef \cond{\CC}{\CC} = \left\{\bfA \in \Mspace{n_1}{\Fq}
    ~|~ \forall \bfc \in \CC,\ \bfc \bfA \in \CC\right\}.
  \]
\end{definition}

\subsubsection{Relation to our problem.}
The first important point is that almost any code of length
$n+\lambda$ has a \textit{trivial right stabilizer}, \ie a stabilizer
of the form $\{\alpha \bfI_{n+\lambda} ~|~ \alpha \in \Fq\}$.
However, the stabilizer of $\Lambda_i(\Cpub)$ is non trivial, since it
contains the matrices \eqref{eq:idempotents}.

The second point is that $\stabr{\Lambda_i(\Cpub)}$ can be computed by
solving a linear system. In general, given a parity--check matrix $\bfH$
for $\CC$, the elements of $\stabr{\CC}$ are nothing but the solutions
$\bfM \in \Mspace{n+\lambda}{\Fq}$ of the system
\begin{equation}\label{eq:compute_stab}
  \bfG \bfM \bfH^{\top} = \bf0.
\end{equation}

\subsubsection{Idempotents and decomposition of the identity.}
The matrices $\bfE_1$ and $\bfE_2$ of (\ref{eq:idempotents}) are {\em
  idempotents} of the right stabilizer algebra of $\Lambda_i(\Cpub)$,
{\ie} elements satisfying $\bfE_1^2 = \bfE_1$ and $\bfE_2^2 = \bfE_2$.
In addition, they provide what is usually called {\em a decomposition
  of the identity with orthogonal idempotents}. The general definition
is given below.
\begin{definition}\label{def:decomp_Id}
  In a matrix algebra $\mathcal A \subseteq \Mspace{n}{\Fq}$, a tuple
  $\bfE_1, \dots, \bfE_r$ of nonzero idempotents are said to be a {\em
    decomposition of the identity into orthogonal idempotents} if they
  satisfy,
\[
  \forall 1 \leq i,j \leq r,\ \bfE_i \bfE_j = \bfzero
  \quad \textrm{and} \quad
  \bfE_1 + \cdots + \bfE_r = \bfI.
\]
Such a decomposition is said to be {\em minimal} if none of the $\bfE_i$'s
can be written as a sum of two nonzero orthogonal idempotents.
\end{definition}

\begin{proposition}
  A code $\CC \subseteq \Fqm^n$ is split if and only if $\stabr{\CC}$
  has a nontrivial decomposition of the identity into orthogonal
  idempotents. 
\end{proposition}

\begin{proof}
  Suppose that $\stabr{\CC}$ contains such a decomposition of the
  identity into orthogonal idempotents
  $\bfI = \bfE_1 + \cdots + \bfE_r$. Since the $\bfE_i$'s
  commute pairwise and are diagonalizable (indeed, being idempotent,
  they all cancelled by the split polynomial $X^2 - X$), they are
  simultaneously diagonalizable. Thus, there exists
  $\bfQ \in \GL{n}{\Fq}$ such that
    \[
     \bfE_1  =
    \bfQ^{-1}\begin{pmatrix}
      \bfI_{n_1}  & & (0)\\
      &  \ddots & \\
      (0) & & (0)
    \end{pmatrix}\bfQ
    , \dots,\ 
    \bfE_r =
    \bfQ^{-1}\begin{pmatrix}
      (0)  & & (0)\\
      &  \ddots & \\
      (0) & & \bfI_{n_r}
    \end{pmatrix}\bfQ,
  \]
  for some positive integers $n_1, \dots, n_r$ such that
  $n_1 + \cdots + n_r = n$.

The code $\CC' = \CC \bfQ$ has the matrices
  \begin{equation}\label{eq:diag_idempotents}
  \bfE'_1 = 
    \begin{pmatrix}
      \bfI_{n_1}  & & (0)\\
      &  \ddots & \\
      (0) & & (0)
    \end{pmatrix}
    , \dots,\ \bfE'_r = 
    \begin{pmatrix}
      (0)  & & (0)\\
      &  \ddots & \\
      (0) & & \bfI_{n_r}
    \end{pmatrix}
  \end{equation}
  in its right stabilizer algebra, and one can easily check that
  $\CC' = \CC' \bfE'_1 \oplus \cdots \oplus \CC' \bfE'_r$, leading to a
  block--wise generator matrix of $\CC'$.
  Thus, $\CC$ has a generator matrix of the form
  \begin{equation}\label{eq:decomp_G}
    \begin{pmatrix}
      \bfG_1 & & (0) \\ & \ddots & \\ (0) & & \bfG_r 
    \end{pmatrix} \bfQ^{-1}.
  \end{equation}
  Conversely, if $\CC$ has a generator matrix as in (\ref{eq:decomp_G}),
  one can easily deduce a decomposition of the identity in $\stabr{\CC}$ into the idempotents (\ref{eq:diag_idempotents}). \qed
\end{proof}

In particular, a code is indecomposable if and only if its right
stabilizer algebra has no nontrivial idempotent. Such an algebra is said
to be {\em local}.

A crucial aspect of minimal decompositions of the identity is the
following, sometimes referred to as the Krull--Schmidt Theorem.

\begin{theorem}[{\cite[Thm.~3.4.1]{DK94}}]\label{thm:Krull-Schmidt}
  Let $\mathcal A \subseteq \Mspace{n}{\Fq}$ be a matrix algebra and
  $\bfE_1, \dots, \bfE_r$ and $\bfF_1, \dots,$ $\bfF_s$ be two minimal
  decompositions of the identity into orthogonal idempotents.  Then,
  $r = s$ and there exists $\bfA \in \mathcal{A}^{\times}$ such that,
  after possibly re-indexing the $\bfF_i$'s, we have
  $\bfF_i = \bfA \bfE_i \bfA^{-1}$, for any $i \in \{1, \dots, s\}$.
\end{theorem}

In short: a minimal decomposition of the identity into idempotents is
unique up to conjugation.

\subsubsection{Algorithmic aspects.} 
Given a matrix algebra, a decomposition of the identity into minimal
idempotents can be efficiently computed using Friedl and Rony\'ai's
algorithms \cite{FR85,R90}. Such a calculation is presented in the
case of stabilizer algebras of codes in \cite{CDG20}.  Further, in
\S~\ref{subsec:discussion}, we present the calculation in a simple case
which turns out to be the generic situation for our cryptanalysis.

\subsection{Description of our extension of Overbeck's
  attack}\label{subsec:descr_attack_long}
The attack summarizes as follows.
Recall that the public code $\Cpub$ has a generator matrix
\[
  \Gpub = (\bfX ~|~ \Gsec) \bfP.
\]

\medskip

\noindent {\bf Step 1.}
Compute $i$ so that the code $\Lambda_i(\Cpub)$ splits as in
(\ref{eq:gm_Lambda_i}), \ie has the shape
\begin{equation}\label{eq:recall_decomp_mat}
  \begin{pmatrix}
    \bfY & \bfzero \\ \bfzero & \bfG_i
  \end{pmatrix}\bfP
\end{equation}
where $\bfG_i$ is a generator matrix of
$\Lambda_i(\Csec)$ and $\bfP$ the column scrambler.  In the sequel,
{\bf we suppose that $\Lambda_i(\Csec)$ is indecomposable}.
This assumption is discussed further in \S~\ref{subsec:discussion}.



\medskip
  
\noindent {\bf Step 2.} Compute $\stabr{\Lambda_i (\Cpub)}$.
We know that this algebra contains the matrices
\begin{equation}\label{eq:E1_E2}
  \bfE_1 = \bfP^{-1}
  \begin{pmatrix}
    \bfI_{\lambda} & \bfzero \\ \bfzero & \bfzero
  \end{pmatrix}
  \bfP
  \quad
  \textrm{and}
  \quad
  \bfE_2 = \bfP^{-1}
  \begin{pmatrix}
    \bfzero & \bfzero \\ \bfzero & \bfI_n
  \end{pmatrix}
  \bfP.
\end{equation}
Next, using the algorithms described in \cite{FR85,R90,CDG20},
compute a minimal decomposition of the identity of
$\stabr{\Lambda_i (\Cpub)}$ into orthogonal idempotents. The following
statement relates any such minimal decomposition to the matrices $\bfE_1$
and $\bfE_2$ in \eqref{eq:E1_E2}.

\begin{lemma}\label{lem:the_good_idempotent}
  Assume that $\lambda < n$. Under the assumption that
  $\Lambda_i(\Csec)$ is an indecomposable code, any minimal
  decomposition of the identity into orthogonal idempotents in
  $\Lambda_i (\Cpub)$ contains a unique element $\bfF$ of rank $n$.
  Moreover, there exists $\bfA \in \stabr{\Lambda_i(\Cpub)}^{\times}$
  such that $\bfF = \bfA^{-1} \bfE_2 \bfA$ where $\bfE_2$ is the
  matrix introduced in (\ref{eq:E1_E2}).
\end{lemma}

\begin{proof}
  Consider the pair $\bfE_1, \bfE_2$ introduced in
  (\ref{eq:E1_E2}). The matrix $\bfE_2$ has rank $n$ and projects the
  code $\Lambda_i(\Cpub)$ onto the code with generator matrix
  $(\bfzero\ \bfG_i) \bfP$, where $\bfG_i$ is a generator matrix of
  $\Lambda_i(\Csec)$. Since $\Lambda_i(\Cpub)$ is supposed to be
  indecomposable, $\bfE_2$ cannot split into
  $\bfE_2 = \bfE_{21} + \bfE_{22}$ such that
  $\bfE_{21} \bfE_{22} = \bfE_{22} \bfE_{21} = \bfzero$, since this
  would contradict the indecomposability of $\Lambda_i (\Csec)$.
  Next, either $\bfE_1, \bfE_2$ is a minimal decomposition or,
  $\bfE_1$ splits into a sum of orthogonal idempotents (if the code
  with generator matrix $\bfY$ splits). In the latter situation, one
  deduces a minimal decomposition of the identity of the form
  $\bfE_{11}, \dots, \bfE_{1r}, \bfE_2$. Now,
  Theorem~\ref{thm:Krull-Schmidt}, permits to conclude that any other
  minimal decomposition is conjugate to the previous one and hence
  contains a unique element of rank $n$ which is conjugate with
  $\bfE_2$. \qed
\end{proof}

\medskip

\noindent {\bf Step 3.} Once we have computed a minimal decomposition
of the identity into minimal idempotents, according to
Lemma~\ref{lem:the_good_idempotent} and Theorem~\ref{thm:Krull-Schmidt}, we have computed
$\bfF \in \stabr{\Lambda_i(\Cpub)}$ of rank $n$ satisfying
$\bfF = \bfA^{-1} \bfE_2 \bfA$ for some unknown matrix
$\bfA \in \stabr{\Lambda_i(\Cpub)}^\times$.

\begin{proposition}\label{prop:almost_the_secret}
  The code, $\Cpub \bfF$ is contained in the code with generator
  matrix
  \[ \left(\bfzero ~|~ \Gisec
    \right) \bfP \bfA,
  \]
  where $\Gisec$ is a generator matrix of the code $\Cisec$ introduced
  in Definition~\ref{def:ct_cond}.  
\end{proposition}

Before proving the previous statement, let us discuss it quickly.  The
result may seem disappointing since, even if we discarded the
distortion matrix, we do not recover exactly the secret code.
However,
\begin{enumerate}
\item the approach is relevant for small $i$'s, and if $i\leq t$,
  where $t$ is the rank of the error term in the encryption process,
  then, the algorithm described in \S~\ref{sec:decodingTwisted}
  decodes $\Ctsec$ (and hence $\Cisec$ since it is contained in
  $\Ctsec$) as efficiently as $\Csec$ itself.
\item In \S~\ref{subsec:discussion}, we provide some heuristic
  claiming that, most of the time, $\Cpub \bfF$ is nothing but the
  code with generator matrix
  \[
    \left(\bf0 ~|~ \Gsec \right) \bfP \bfA.
  \]
\end{enumerate}

\begin{proof}[of Proposition~\ref{prop:almost_the_secret}]
  Recall that $\bfF = \bfA^{-1} \bfE_2 \bfA$ for some matrix
  $\bfA \in \stabr{\Lambda_i(\Cpub)}$. Then, since $\bfA$ is invertible,
  we deduce that $\Lambda_i(\Cpub) \bfA^{-1} = \Lambda_i(\Cpub)$.
  Therefore,
  \[
    \Lambda_i(\Cpub) \bfF = \Lambda_i(\Cpub) \bfE_2 \bfA.
  \]
  From~(\ref{eq:recall_decomp_mat}) and (\ref{eq:E1_E2}), the code
  $\Lambda_i (\Cpub)\bfE_2$ has a generator matrix of the form
  \( (\bfzero ~|~ \bfG_i)\bfP\) and hence the code
  $\Lambda_i(\Cpub) \bfF$ has a generator matrix
  \begin{equation}\label{eq:gen_mat_LCpubF}
    (\bfzero ~|~ \bfG_i) \bfP \bfA.
  \end{equation}
  Next, the code $\Cpub$ is
  contained in $\Lambda_i(\Cpub)$ but also in
  $\tclosure{\Lambda_i(\Cpub)}{i}$. Moreover, according to
  Remark~\ref{rem:alter_closure}, we have
  \[
    \Cpub \subseteq \Cipub = \bigcap_{j = 0}^i {\left(\Lambda_i(\Cpub)
      \right)}^{[-j]}.
  \]
  Since both $\bfP$ and $\bfA$ have their entries in $\Fq$, they commute with
  the operations of raising to any $q$--th power
  and we deduce that
  \[\Cpub \bfF \subseteq \tclosure{\Lambda_i(\Cpub)}{i}\bfF. \]
  Then, from~(\ref{eq:gen_mat_LCpubF}), we deduce that $\Cpub \bfF$
  is contained in the code with generator matrix
  \[
    \left(\bfzero ~|~ \Gisec\right) \bfP \bfA.
  \]\qed
\end{proof}

\medskip

\noindent {\bf Step 5.} With the previous results at hand, given a
ciphertext $\bfy = \bfm \Gpub + \bfe$ with $\rank(\bfe) \leq t$, we can
compute
\[
  \bfy \bfF= \bfm \Gpub \bfF + \bfe \bfF.
\]
Then, we remove its $\lambda$ leftmost entries.  Since $\bfF$ has its
entries in $\Fq$, $\rank(\bfe \bfF)\leq \rank(\bfe)$.  Next,
$\bfm \Gpub \bfF$ with the $\lambda$ leftmost entries removed is a
codeword in $\tclosure{\Lambda_i(\Csec)}{i}$  which can be decoded using
the algorithm introduced in \ref{sec:decodingTwisted}. This yields the
plaintext $\bfm$.
 
\subsection{Summary of the attack}
According to the description in \S~\ref{subsec:descr_attack_long}, the
attack is now summarized in Algorithm~\ref{algo:attack} below.

\begin{algorithm}
  \DontPrintSemicolon
  \KwIn{$\Gpub$, a ciphertext $\bfy$ and the rank of the error term $t$}
  \KwOut{A pair $(\bfm \bfe) \in \Fqm^{k}\times \Fqm^n$
  such that $\rank(\bfe) = t$ and $\bfy = \bfm \Gpub +\bfe$ or `?' if fails}

\medskip
  Compute a generator matrix of $\Lambda_i (\Cpub)$ for the least $i$
  for which the code splits.  \;
  \If{no such $i$ exists}{Return `?'}

  Compute a minimal decomposition of the identity of
  $\stabr{\Lambda_i(\CC)}$ and extract its unique term $\bfF$ of
  rank $n$.\;
  \If{no such $\bfF$ exists}{Return `?'}

  Compute $\bfy \bfF$ and apply to it the decoder described in
  \ref{sec:decodingTwisted}.\;

  {return the output $\bfm$ of the decoder (possibly `?' if the
    decoder fails)}.\;
	\caption{Summary of the attack}\label{algo:attack}
\end{algorithm}

\subsection{Discussions and simplifications}\label{subsec:discussion}
For the attack presented in Algorithm~\ref{algo:attack} to work,
several assumptions are made. Here we discuss these assumptions and
their rationale. We also point out that in our specific case, the
algebra $\stabr {\Lambda_i(\Cpub)}$ will be very specific. This may
permit to avoid to consider the difficult cases of Friedl  Rony\'ai's
algorithms.

\subsubsection{Indecomposability of $\Lambda_i(\Csec)$.}
An important assumption for the attack to succeed is that
$\Lambda_i(\Csec)$ does not split. Note first that in the classical GPT case,
 $\Csec$ is a Gabidulin code. And so, this always holds as soon as $i < n-k$.
 
 This is a consequence of the following statement and the fact that if
 $\Csec$ is a Gabidulin code, and so for any $i > 0$, also
 $\Lambda_i(\Csec)$ is a Gabidulin code. Thus, according to the
 following statement it is indecomposable.

\begin{proposition}
  An MRD code $\CC \varsubsetneq \Fqm$ never splits.
\end{proposition}

\begin{proof}
  Let $\CC \subseteq \Fqm^n$ be an MRD code of dimension $k$. Suppose
  it splits into a direct sum of two codes $\CC_1, \CC_2$ of
  respective lengths $n_1, n_2$ and dimensions $k_1, k_2$.
  Then, $\CC_1$ has codewords of rank weight $n_1 - k_1 + 1$
  and $\CC_2$ has words of weight $n_2 - k_2 + 1$.
  Such words are also words of $\CC$ and, since $\CC$ is MRD, we have
  \begin{eqnarray*}
    n_1 - k_1 + 1 & \geq & n - k + 1 \\
    n_2 - k_2 + 1 & \geq & n - k + 1 
  \end{eqnarray*}
  Summing up these two inequalities and using the fact that
  $n_1+n_2 = n$ and $k_1+k_2 = k$, we get a contradiction.\qed
\end{proof}

In the general case of twisted Gabidulin codes the situation is more
complicated. However, twisted Gabidulin codes are contained in Gabidulin
codes of larger dimensions, hence so are their images by the
$\Lambda_i$ operator. It seems very unlikely that a Gabidulin code
could contain large subcodes that split.

\subsubsection{On the structure of $\stabr{\Lambda_i(\Cpub)}$.}
A crucial step of the attack is the computation of a decomposition of
the identity of $\stabr{\Lambda_i(\Cpub)}$ into a sum of orthogonal idempotents.
For this, we referred to Friedl Rony\'ai \cite{FR85,R90}. Actually, our setting
is rather specific and the structure of this stabilizer algebra is pretty well
understood. Let us start with a proposition.




\begin{proposition}\label{prop:stab_split}
  Let $\CC$ be an $\Fqm$--linear code of length $n + \lambda$ and
  dimension $K$ with a generator matrix of the
  shape~(\ref{eq:gm_Lambda_i}), {\em i.e.}
  \[
    \begin{pmatrix}
      \bfG_1 & \bf0 \\  \bf0 & \bfG_2
    \end{pmatrix},
  \]
  with $\bfG_1 \in \Mspace{k_1,\lambda}{\Fqm}$ for some integer $k_1$
  and $\bfG_2 \in \Mspace{k_2, n}{\Fqm}$ for some integer $k_2$ so
  that $k_1 + k_2 = K$.  Denote by $\CC_1$ and $\CC_2$ the codes with
  respective generator matrices $\bfG_1$ and $\bfG_2$. Then any
  $\bfM \in \stabr{\CC}$ has the shape
  \[
    \bfM =
    \begin{pmatrix}
      \bfA & \bfB \\ \bfC & \bfD
    \end{pmatrix},
  \]
  where $\bfA \in \stabr{\CC_1}$,
  $\bfB \in \cond{\CC_2}{\CC_1}$,
  $\bfC \in \cond{\CC_1}{\CC_2}$ and
  $\bfD \in \stabr{\CC_2}$. \qed
\end{proposition}

\begin{proof}
  Let $\bfc_1 \in \CC_1$, then $(\bfc_1\ \bf0) \in \CC$ and by
  definition of $\bfM$,
  $(\bfc_1\ \bf0) \bfM = (\bfc_1 \bfA\ \bfc_1 \bfB) \in \CC$.  By definition
  of $\CC$, we have $\bfc_1 \bfA \in \CC_1$ and
  $\bfc_1\bfB \in \CC_2$. Since the previous assertions hold for any
  $\bfc_1 \in \CC_1$, then we deduce that $\bfA \in \stabr{\CC_1}$
  and $\bfB \in \cond{\CC_1}{\CC_2}$.
  
  The result for $\bfC, \bfD$ is obtained in the same way
  by considering $(\bf0\ \bfc_2) \bfM$ for $\bfc_2 \in \CC_2$. \qed
\end{proof}

Consequently considering the generator
matrix~(\ref{eq:recall_decomp_mat}) of $\Lambda_i(\Cpub)$, elements of
$\stabr{\Lambda_i(\Cpub)}$ have the shape
\begin{equation}\label{eq:shape_of_stab}
  \begin{pmatrix}
    \bfA & \bfB \\ \bfC & \bfD
  \end{pmatrix},
\end{equation}
where $\bfA \in \stabr{\CC_{\bfY}}$ ($\CC_{\bfY}$ being the code with
generator matrix $\bfY$),
$\bfB \in \cond{\Lambda_i(\Csec)}{\CC_{\bfY}}$,
$\bfC \in \cond{\CC_{\bfY}}{\Lambda_i(\Csec)}$ and
$\bfD \in \stabr{\Lambda_i(\Csec)}$.

Here again, we claim that  is very likely that the stabilizer algebras
of $\CC_{\bfY}$ and $\Lambda_i(\Cpub)$ are trivial, \ie contain only scalar
multiples of the identity matrix and that the conductors
$\cond{\CC_{\bfY}}{\Lambda_i(\Csec)}$ and $\stabr{\Lambda_i(\Csec)}$
are zero. This claim is discussed further in \S~\ref{ss:ev2}.

In such a situation, we have:
\begin{equation}\label{eq:stabr_struct}
  \stabr{\Lambda_i(\Cpub)} = \left\{
    \bfP^{-1}
    \begin{pmatrix}
      a \bfI_{\lambda} & \bfzero \\ \bfzero & b \bfI_n
    \end{pmatrix}\bfP
    ~\bigg|~ a,b \in \Fq
  \right\}.
\end{equation}
Hence this algebra has dimension $2$ and the calculation of the matrix
\begin{equation}\label{eq:the_very_good_idempotent}
\bfP^{-1}     \begin{pmatrix}
     \bfzero & \bfzero \\ \bfzero & \bfI_n
    \end{pmatrix}\bfP
  \end{equation}
  can be performed as follows.

  \begin{enumerate}
  \item First extract a singular matrix of
    $\stabr{\Lambda_i(\Cpub)}$. For that, take $\bfU, \bfV$ a basis of
    $\stabr{\Lambda_i(\Cpub)}$. If $\bfV$ is singular we are
    done. Otherwise, compute a root of the univariate polynomial
    $\det (\bfU + X \bfV)$. This yields a singular element $\bfR$ of
    $\stabr{\Lambda_i(\Cpub)}$ corresponding either to $a = 0$ or
    $b = 0$ in the description~(\ref{eq:stabr_struct}).

\item  Next, rescale $\bfR$ as $\nu \bfR$ in order to get an idempotent
  element.  If the obtained idempotent has rank $n$ set
  $\bfF = \nu \bfR$, otherwise (it will have rank $\lambda$), set
  $\bfF = \bfI_{n+\lambda} - \nu \bfR$.
\end{enumerate}

The obtained matrix $\bfF$ is nothing but the target
matrix in (\ref{eq:the_very_good_idempotent}). Therefore, one can even
skip the proof of Proposition~\ref{prop:almost_the_secret} and observe
that the code $\Cpub \bfF$ will be {\bf exactly} the code with
generator matrix
\[
  \left( \bfzero ~|~ \Gsec \right) \bfP.
\]

\subsection{Complexity}
Considering the previous simple case which remains very likely, we analyze the
cost of the various computation steps.
\begin{itemize}
\item The computation of $\Lambda_i(\Cpub)$ can be done by iterating
  $i$ successive Gaussian eliminations (we assume that raising an
  element of $\Fqm$ to the $q$--th power can be done for free, for
  instance by representing $\Fqm$ with a normal basis).
  Thus, a cost $O(i n^{\omega})$ operations in $\Fqm$ and hence
  $O(i m^2 n^\omega)$ operations in $\Fq$. Here, $\omega$
  denotes the usual exponent for the cost of the product of two
  $n \times n$ matrices.
\item The computation of $\stabr{\Lambda_i(\Cpub)}$ is done by solving
  the linear system (\ref{eq:compute_stab}). The system has $n^2$
  unknowns in $\Fq$ and $k_i (n-k_i) = O(n^2)$ equations in $\Fqm$ and
  hence $O(mn^2)$ equations in $\Fq$. This yields a cost of
  $O(m n^{2\omega})$ operations in $\Fq$ (see
  \cite[Thm.~8.6]{BCGLLSS17} for the complexity of the resolution of a
  non square linear system).
\end{itemize}
In the aforementioned simple case, the remaining operations are negligible
compared to the calculation of the stabilizer algebra, which turns out
to be the bottleneck of the calculation. This overall cost is hence in
\[
  O(mn^{2\omega})\ \ \textrm{operations\ in\ }\Fq.
\]

\subsection{Discussion about the claims on conductors and
  stabilizers}\label{ss:ev2}
Back to the description (\ref{eq:shape_of_stab}) of the elements of
$\stabr{\Lambda_i(\Cpub)}$. Let us discuss the validity of the claim.

\subsubsection{Conductors are likely to be zero.}
Let $\bfC \in \cond{\CC_{\bfY}}{\Lambda_i(\Csec)}$, then the code
$\CC_{\bfY} \bfC$ is a subcode of $\Lambda_i(\Csec)$ and one proves
easily that any element of $\CC_{\bfY} \bfC$ has a row support
contained in the row space of $\bfC$. Since
$\bfC \in \Mspace{\lambda, n}{\Fq}$, its rank is at most equal to
$\lambda$ and hence the code $\CC_{\bfY} \bfC$ has a row space
contained in a space of dimension $\leq \lambda$.  It seems unlikely
that the code $\Lambda_t(\Csec)$ contains such a space. In particular,
this cannot happen if the minimum distance of $\Lambda_i(\Csec)$
exceeds $\lambda$.

Now, consider $\bfB \in \cond{\Lambda_i(\Csec)}{\CC_{\bfY}}$. Suppose
first that $\bfB$ has full rank.  Since
$\dim (\Lambda_i(\Csec))\gg \lambda$, the code
$\Lambda_i (\Csec) \bfB$ is likely to be equal to $\Fqm^\lambda$ and
hence cannot be contained in $\CC_{\bfY}$, a contradiction.  If $\bfB$
has not full rank, then, the code $\Lambda_i (\Csec) \bfB$ is likely
to be equal to the subspace of $\Fqm^\lambda$ of all the vectors whose
row support is in the row space of $\bfB$ and we can assume that
$\CC_{\bfY}$ has no such subspace. Indeed, if it did, it would entail
that $\CC_{\bfY}^\perp$ (and hence $\Lambda_i(\Cpub)^\perp$ too) would
have a parity-check matrix of the form $(\bf0 ~|~ \bfH')(\bfP^{-1})^{\top}$
as in \eqref{eq:zeroblock}. Details are left to the reader.

\subsubsection{Stabilizers are likely restrict to scalar matrices.}
For $\CC_{\bfY}$, this code is close to be random and random codes
have trivial stabilizer algebras with a high probability.

For $\Lambda_i(\Csec)$ the right stabilizer algebra might be a larger
one.  Indeed, regarding the proof of
Proposition~\ref{prop:distinguish_twisted} (see \cite[Thm.~4]{PRW18})
we can see that $\Lambda_{t}(\CC)$ is a code generated by the
evaluations of $q$-monomials and such a code, when $n=m$ has a right
stabilizer algebra equal to a matrix representation of $\Fqm$. This is
a consequence of the fact that an $\Fqm$--space spanned by
$q$--monomials is $\Fqm$--linear on the left but also on the right.
Thus, $\stabr{\Lambda_i(\Csec)}$ might be such a larger algebra. In
this situation, the calculation of a decomposition of the identity into
orthogonal idempotents is slightly more complicated but remains
definitely possible in polynomial time using Friedl Rony\'ai algorithms.

\section{Don't twist again}{\label{sec:dontTwist}}
In this section we first show that, even for twisted Gabidulin codes,
the application of the $q$-sum operator allows to \emph{distinguish}
them from random codes. It is therefore possible to apply the attack
described in \S~\ref{sec:extension} to the GPT cryptosystem
instantiated with these codes.  In the first part of this section we
discuss the behaviour of raw twisted Gabidulin codes with respect to
the operator $\Lambda_i$ or equivalently, how the use of $\Lambda_i$
allows to distinguish them from random codes.  In the second part, we focus on $q$-operator applied to the corresponding
public key and we will prove that even in this case, we have a
generator matrix with a structure similar to
\eqref{eq:Lambda1XLambda1G} and that the corresponding codes
split. This allows us to apply the results of \S~\ref{sec:extension}.

\subsection{A distinguisher}\label{subsec:dist_twisted}
First, recall Propositions~\ref{prop:distinguish_twisted} and~\ref{prop:Lambdai(C)} about the dimension of the $q$-sum
operator applied respectively to twisted Gabidulin codes and to random
codes. In particular, recall that if $\CC$ is a random code, $\dim(\Lambda_i(\CC))=(i+1)k$ with high probability. Then, we remark that, if $i<\frac{n-k-\ell}{\ell+1}$

\begin{eqnarray}
  \dim(\Lambda_i(\TGab{\bfg}{\bft}{\bfh}{\bfeta}[n,k]))&=&k+i+\ell(i+1)<(i+1)k
  =\dim(\Lambda_i(\CC))\\ &\Longleftrightarrow& i>\ell/(k-\ell-1), \label{eq:distinguisherTwisted}
 \end{eqnarray}
where $\TGab{\bfg}{\bft}{\bfh}{\bfeta}[n,k]$ is a twisted Gabidulin
code (see \S~\ref{subsec:twisted}).

Thus, the inequality $i>\ell/(k-\ell-1)$  is satisfied by any positive $i$, if $k>2\ell+1$. We
notice that this is often the case if we consider a small number of
twists as in Table~\ref{tab:twistedPar}. This means that, even if the dimension of the $q$-sum applied to these codes is
greater than that of the $q$-sum of a Gabidulin code, we can however
still distinguish them for random codes.

Thus, this distinguisher can be exploited to construct an attack
against the GPT cryptosystem instantiated with twisted Gabidulin
codes, instead of classical ones.


\subsection{The structure of $\Lambda_i(\GG_{Tpub})$}\label{subsec:structureLambdaiGtpub}
From now on, we consider the GPT cryptosystem instantiated with a
twisted Gabidulin code $\TGab{\bfg}{\bft}{\bfh}{\bfeta}[n,k]$ with the
parameters defined in Assumption~\ref{ass:parameters}. We denote by $\Gtpub$
the corresponding public key, obtained as \eqref{eq:Gpub} by just
replacing $\Gsec$ with a generator matrix $\bfGt$ (of the form \eqref{eq:
  twistedGenMatrix}) of the code
$\TGab{\bfg}{\bft}{\bfh}{\bfeta}[n,k]$ and by $\GG_{\text{Tpub}}$ the linear
code which has $\Gtpub$ as generator matrix. Again, as for the Gabidulin codes scheme, we can discard the matrix $\bfS$.

We now apply the $q$-sum operator to $\Gtpub$, and as
\eqref{eq:Lambdai(Gpub)}, we get
\[
\Lambda_i(\Gtpub)=[\Lambda_i(\bfX)|\Lambda_i(\bfGt)]\bfP,
\]
where
$\bfP\in \GL{n+\lambda}{\Fq}$ is the column scrambler.

Let $i<\frac{n-\ell-k}{l+1}$ and write $\bfX$ (as in \S~\ref{subsec:structureLambdaiGpub}) according to its rows.

Now, for simplicity we consider  that $\ell=1$, $\eta_1=1$ and $i=1$. Recall that the structure of $\bfGt$
is given in \eqref{eq: twistedGenMatrix}. Then, we have
\[
(\Lambda_1(\bfX)~|~\Lambda_1(\bfGt))=\left(\begin{array}{c|c}
	\bfx_0 &\bfg\\
	\bfx_1 &\bfg^{[1]}\\
	\vdots &\vdots\\
	\bfx_{h_1} &\bfg^{[h_1]}+\bfg^{[k-1+t_1]}\\
	\vdots &\vdots\\
	\bfx_{k-1} &\bfg^{[k-1]}\\
	\hline
	\bfx_0^{[1]} &\bfg^{[1]}\\
	\bfx_1^{[1]} &\bfg^{[2]}\\
	\vdots &\vdots\\
	\bfx_{h_1-1}^{[1]} &\bfg^{[h_1]}\\
	\bfx_{h_1}^{[1]} &\bfg^{[h_1+1]}+\bfg^{[k+t_1]}\\
	\vdots &\vdots\\
	\bfx_{k-2}^{[1]} &\bfg^{[k-1]}\\
	\bfx_{k-1}^{[1]} &\bfg^{[k]}\\
	
\end{array}\right)
\longrightarrow
\left(\begin{array}{c|c}
	\bfx_0 &\bfg\\
	\bfx_1 &\bfg^{[1]}\\
	\vdots &\vdots\\
	\bfx_{h_1-1} &\bfg^{[h_1-1]}\\
	\bfx_{h_1-1}^{[1]} & \bfg^{[h_1]}\\
	\bfx_{h_1+1} & \bfg^{[h_1+1]}\\
	\vdots &\vdots\\
	\bfx_{k-1} &\bfg^{[k-1]}\\
	\bfx_{k-1}^{[1]} &\bfg^{[k]}\\
	\hline
	\bfx_0^{[1]} &\bfg^{[1]}\\
	\bfx_1^{[1]} &\bfg^{[2]}\\
	\vdots &\vdots\\
	\bfx_{h_1-2}^{[1]} &\bfg^{[h_1-1]}\\
	\bfx_{h_1} &\bfg^{[h_1]}+\bfg^{[k-1+t_1]}\\
	\bfx_{h_1}^{[1]} &\bfg^{[h_1+1]}+\bfg^{[k+t_1]}\\
	\bfx_{h_1+1}^{[1]} &\bfg^{[h_1+2]}\\
	\vdots &\vdots\\
	\bfx_{k-2}^{[1]} &\bfg^{[k-1]}\\
\end{array}\right)
\]
where the second matrix is obtained by permuting the rows of the first
one. We now observe that the first block of the second matrix can be
rewritten as $[\tilde{\bfX'}|\Moore{k+1}{\bfg}]$ and
so, after performing row elimination, we get
\[
\left(\begin{array}{c|c}
	\tilde{\bfX'} & \Moore{k+1}{\bfg}\\
	\hline
	\bfx_{h_1}-\bfx_{h_1-1}^{[1]} &\bfg^{[k-1+t_1]}\\
	\bfx_{h_1}^{[1]}-\bfx_{h_1+1} &\bfg^{[k+t_1]}\\
	\hline
	\bfx_0^{[1]}-\bfx_1 &\bfzero\\
	\bfx_1^{[1]}-\bfx_2 &\bfzero\\
	\vdots &\vdots\\
	\bfx_{h_1-2}^{[1]}-\bfx_{h_1-1} &\bfzero\\
	\bfx_{h_1+1}^{[1]}-\bfx_{h_1+2} &\bfzero\\
	\vdots &\vdots\\
	\bfx_{k-2}^{[1]}-\bfx_{k-1} &\bfzero\\
\end{array}\right)
\]
Therefore, we have the following result.
\begin{lemma}\label{lem:lambdaiGtpub}
	Let $i<\frac{n-\ell-k}{\ell+1}$. Then, up to row elimination
			\begin{equation}\label{eq:LambdaiXLambdaiGTw}
	(\Lambda_i(\bfX)~|~\Lambda_i(\bfG_{T}))=	\begin{pmatrix}
			\bfY & \Lambda_i(\bfGt)\\
			\tilde{\bfX} &\bfzero\\
		\end{pmatrix}
	\end{equation}
	where, 
	\[
	\tilde{\bfX}= \left\{
	\begin{array}{ccccl}
		\begin{pmatrix}
			\bfX_T''\\
		\end{pmatrix}&\in& \Mspace{k-1-2\ell}{\Fqm}
		& \text{if} & i = 1\\
		\begin{pmatrix}
			\Lambda_{i-1}(\bfX_T'')\\
			\bfX'''
		\end{pmatrix}&\in& \Mspace{i(k-1-2\ell)+(i-1)\ell}{\Fqm} & \text{if} & i > 1
	\end{array}
	\right.
	\]
	$\bfY \in \Mspace{k+i+\ell(i+1), \lambda}{\Fqm}$ and the matrix $\bfX_T''$ is defined as,
	\begin{equation}\label{eq:Xt}
		\bfX_T''= \bfX^{[1]}_{\{0, \ldots, k-2\}\setminus\{h_i-1, h_i\mid 1\leq i\leq  \ell\}}-\bfX_{\{1, \ldots, k-1\}\setminus \{h_i, h_i+1\mid 1\leq i \leq \ell\}},
	\end{equation}
	where
	$\bfX^{[1]}_{\{0, \ldots, k-2\}\setminus\{h_i-1, h_i\mid 1\leq i\leq
		\ell\}}$ is a submatrix of $\bfX^{[1]}$ composed by the first $k-1$
	rows except all the $(h_i-1)$-th, $h_i$-th rows and
	$\bfX_{\{1, \ldots, k-1\}\setminus \{h_i, h_i+1\mid 1\leq i \leq \ell\}}$
	is a submatrix of  $\bfX$ determined by all the
	rows, starting from the second one, except the $h_i$-th,
	$h_i+1$-th ones. Finally, $\bfX'''\in \Mspace{i-1, \lambda}{\Fqm}$.
\end{lemma}

\begin{proof}
	Using the same elimination techniques as before, we can extend the proof to the case $\ell>1$, $\bfeta\in (\Fqm\setminus\{0\})^\ell$ and $i>1$. \qed
\end{proof}

Even in this case, we show that it suffices to consider $i=1$ to attack the corresponding GPT scheme.
\subsection{Attacking the system for small i's.}\label{subsec:attackSmalliTw}
We now consider $i=1$. Then by Lemma~\ref{lem:lambdaiGtpub}, $(\Lambda_i(\bfX)~|~\Lambda_i(\bfG_{T}))$ can be transformed into
\[
\begin{pmatrix}
		\bfY & \Lambda_1(\bfGt)\\
		\bfX_{T}''&\bfzero\\
	\end{pmatrix}
\]
As in \S~\ref{subsec:attackForSmallI} (see Lemma~\ref{lem:LambdaiGpubblocks}), under some assumptions on the parameters, we can split the previous matrix into two blocks.
\begin{lemma}
	If $k\geq 4s+2\ell+1$, then, with a high probability,
	\begin{equation}\label{eq:lamba1GpubTblocks}
	\begin{pmatrix}
		\bf0 & \Lambda_1(\bfGt)\\
		\bfX_{T}''&\bfzero\\
	\end{pmatrix}
	\end{equation}
up to row eliminations.
\end{lemma}
\begin{proof}
  The proof is analogous to the proof of
  Lemma~\ref{lem:LambdaiGpubblocks}.  First we prove that
  $\RowSp_{\Fqm}(\bfX_{T}'')=\RowSp_{\Fqm}(\Lambda_1(\bfX))$ with a
  high probability. Again, we consider the submatrix of $\bfX_{T}''$
  in $\Mspace{\lfloor\frac{k-1-2\ell}{2}\rfloor}{\Fqm}$ obtained by
  alternatively selecting rows of $\bfX_{T}''$. This matrix is
  uniformly random and by Proposition~\ref{prop:Lambdai(C)}, if
  $\frac{k-1-2\ell}{2}\geq 2s$ (which is true by assumption), it has
  rank equal to the rank of $\Lambda_1(\bfX)$ with a high
  probability. Thus the equality
  $\RowSp_{\Fqm}(\bfX_{T}'')=\RowSp_{\Fqm}(\Lambda_1(\bfX))$ holds.
  
  The result follows by noting that
  $\RowSp_{\Fqm}(\bfY)\subseteq \RowSp_{\Fqm}(\Lambda_1(\bfX))$. \qed
\end{proof}
Therefore we can apply the attack of \S~\ref{sec:extension} in order to break the corresponding GPT cryptosystem.

\begin{remark}
Notice that, if $\rank(\bfX)=s\geq \lambda/2$, then $\rank(\bfX_{T})=\lambda$ with high probability and we can apply the Overbeck's attack to this scheme.  In fact, in this case (as in \S~\ref{subsec:Overbeck}), $\dim(\Lambda_1(\GG_{Tpub})^{\perp})=n-k-1-2\ell$, and so the code $\Lambda_1(\GG_{Tpub})$  admits a parity check matrix whose first $\lambda$ columns are $\bf0$. We can then compute a valid column scrambler and attack the system.

\smallskip More generally, we can apply this attack to any
$i< \frac{n-\ell-k}{\ell+1}$ for which \[\rank(\tilde{\bfX})=\lambda,\]
where $\tilde{\bfX}$ is defined in Lemma~\ref{lem:lambdaiGtpub}.
\end{remark}

\section*{Conclusion}
In this paper, we present new observations on the decoding of
Gabidulin codes. These allow us to introduce a decoder for twisted
Gabidulin codes up to a certain threshold, which may be less than half of the
minimum distance.

We then propose an extension of the Overbeck's attack on GPT-like
systems instantiated on Gabidulin or related codes such as twisted
Gabidulin codes.  This attack is efficient as soon as the secret code
$\Lambda_i(\Csec)$ has a small dimension compared to the dimension of $\Lambda_i(\CC)$, where $\CC$ is a random code.  One of the interesting things about our approach is that it succeeds
even when the distortion matrix has a low rank, which might cause the
Overbeck's attack fails. Our attack extension allows to break the proposal of
\cite{PRW18}.




\bibliographystyle{splncs04}
\bibliography{src/codecrypto,src/local_biblio}

\end{document}